\documentclass[12pt, a4paper]{article}
\usepackage[dvips]{color, graphicx}
\usepackage{bm, psfrag, afterpage, color, amsmath, amssymb, latexsym, amsthm, url, lscape}
\usepackage{mathtools}
\usepackage{booktabs}
\usepackage{multirow}
\usepackage{adjustbox}

\renewcommand{\arraystretch}{0.8}

\renewcommand{\arraystretch}{1.0}

\newcommand{\iid}{\stackrel{\mathrm{i.i.d.}}{\sim}}
\newtheorem{theorem}{Theorem}
\newtheorem{proposition}{Proposition}
\newtheorem{lemma}{Lemma}

\def\E{{\rm E}}

\numberwithin{equation}{section}

\usepackage[margin=1in]{geometry}

\usepackage{xcolor}

\usepackage[natbibapa]{apacite}

\usepackage[hidelinks]{hyperref}

\usepackage{algorithm}
\usepackage{algorithmic}

\title{Variable Fusion and Selection via a Spike-and-Slab Approach with Nonlocal Priors}

\author{
  Junya Miyake\textsuperscript{*}, 
  Akira Okazaki\textsuperscript{$\dagger$}, 
  Shuichi Kawano\textsuperscript{$\ddagger$}\\[1ex]
 \small{\textsuperscript{*} Graduate School of Mathematics,  Kyushu University}\\
 \small{\textsuperscript{$\dagger$} The Institute of Statistical Mathematics}\\
 \small{\textsuperscript{$\ddagger$} Faculty of Mathematics,  Kyushu University}
 }

\date{\today}

\begin{document}
\maketitle

\begin{abstract}
Variable fusion in linear regression models is a statistical method that identifies covariates making similar contributions to the response variable and imposes the same coefficient values on them.  
Many methods for variable fusion also incorporate variable selection for practical reasons. 
In this study,  within the Bayesian model averaging (BMA) framework,  we propose a spike-and-slab-based Bayesian method that performs both variable fusion and selection. 
This is challenging in the BMA framework because one must construct a discrete model space that accommodates both selection and fusion and assign suitable priors over that space. 
In the proposed method,  we present a way to explore a model space for variable fusion and selection based on Gibbs sampling by devising a prior distribution for latent variables representing the model. 
Furthermore,  among non-local priors with superior model selection properties,  we construct a prior tailored for variable fusion and use it as the slab distribution. 
We examine the effectiveness of the proposed method through theoretical and empirical studies. 
\end{abstract}

\section{Introduction}
In recent years,  owing to advances in measurement technology and computational resources,  high-dimensional data have come to be routinely acquired and analyzed across a wide range of fields,  including genomics,  medicine,  and finance \citep{buhlmann2016handbook, fan2020statistical}. 
Linear regression models,  because of their high interpretability,  are still widely used even now as a fundamental framework,  and then they continue to attract attention as important models even in high-dimensional settings. 
In high-dimensional settings,  many covariates that do not contribute to the response are included,  and even among the covariates that do contribute,  redundancy often arises.
In particular,  when covariates have a natural order,  covariates that are sequentially close tend to contain similar information and may exhibit similar contributions to the response. 
For example,  in genomics,  measurements at nearby genomic positions often exhibit similar contributions to disease status or drug response \citep{gabriel2002structure, mcinnes2021gwaspgx}. 
In such situations,  variable fusion~---~which involves treating neighboring covariates as the same group and constraining the regression coefficients of explanatory variables within the same group to take the same value~---~becomes an effective option \citep{land1996variablefusion}.

As a frequentist method that encourages variable fusion,  the fused lasso \citep{tibshirani2005fusedlasso} is widely used. 
The fused lasso encourages fusion of variables by imposing an $L_1$ penalty on differences between adjacent regression coefficients,  and then it is likely to yield estimation results  in which the coefficients have the same value. 
This method has the advantage of estimating the fusion structure within a continuous optimization framework.

In a Bayesian framework as well,  there has been research on shrinkage priors that promote equality among neighboring regression coefficients.
For example,  if we place a Laplace-type prior distribution on the difference between regression coefficients corresponding to adjacent covariates,  the difference is strongly pulled toward the vicinity of zero under the posterior distribution. 
As a result,  we can encourage adjacent regression coefficients to become equal. 
This method is called the Bayesian fused lasso \citep{Kyung2010BayesFusedLasso}. 
Furthermore,  the Bayesian fused lasso model that places a horseshoe prior on successive differences of regression coefficients has been also proposed,  which can mitigate over-shrinkage of the differences and lead to more stable estimation \citep{kakikawa2023bayesian}.

As a different Bayesian approach,  research on placing priors on the model space itself and performing model selection using Bayes factors has also received considerable attention  \citep{kass1995bayes}. 
In Bayes-factor-based model selection,  one often not only selects a single best model but also performs Bayesian model averaging (BMA),  which averages estimates over the model space using posterior model probabilities as weights \citep{hoeting1999bma}. 
In problems such as variable selection and variable fusion,  where a large number of discrete candidates arise,  it is particularly useful to avoid a single model and instead to conduct inference by weighting multiple models with their posterior probabilities.

In this study,  we discuss the BMA method based on a spike-and-slab prior to achieve variable selection and variable fusion simultaneously. 
Specifically,  in addition to a \textit{spike distribution for variable selection} that represents regression coefficients being exactly zero,  we introduce a \textit{spike distribution for variable fusion} that represents the differences between adjacent regression coefficients being exactly zero. 
Based on a mixture prior consisting of these two spike distributions and a slab distribution that represents the case where coefficients are neither zero nor fused,  we simultaneously perform variable selection and variable fusion.

While BMA can reflect model uncertainty,  it has been pointed out that,  in finite samples,  non-negligible posterior probability can be assigned even to redundant models that include unnecessary variables \citep{lindley1957statistical}. 
To address this issue,  several prior distributions have been proposed in the variable selection literature with the aim of suppressing the posterior probabilities of redundant models \citep{lindley1957statistical,  bartlett1957comment,  berger1987testing}. 
Representative examples include the product moment (pMOM) prior,  the product inverse moment (piMOM) prior,  and the product exponential moment (peMOM) prior,  all of which share the common property that their probability density is zero in a neighborhood of zero for the regression coefficients \citep{rossell2013hdclassifiers}. 
\citet{johnson2010use} referred to priors with this property as nonlocal priors for variable selection,  and all other priors as local priors. 
Furthermore,  for some nonlocal priors such as the pMOM prior,  it has been theoretically and empirically shown that the posterior probability of redundant models converges faster than under local priors.

In this study,  we extend these insights to the problem of variable fusion. Specifically,  we propose a prior distribution such that the posterior probability decays faster for models in which regression coefficients expected to be fused have the different values. 
We adopt this prior as the slab distribution. 
We show that this distribution retains,  in both aspects of variable selection and variable fusion,  part of the desirable asymptotic properties possessed by the pMOM prior.

In BMA,  it is also important how a prior distribution is specified over the model space \citep{clyde2004model}. 
The BMA based on a spike-and-slab prior can naturally induce a prior distribution over the model space by placing a prior distribution on a latent indicator vector representing the model. 
Since this study simultaneously addresses variable selection and variable fusion,  we introduce a ternary latent indicator vector corresponding to two spike distributions and one slab distribution. 
We propose to assign a Markov-chain-based prior distribution to this ternary latent indicator vector. 
By setting the initial distribution and transition probabilities appropriately,  our Markov-chain-based prior distribution enables us to induce a uniform prior distribution over the model space. 
We also provide a sampling scheme from our proposed Bayesian model.  

The main contributions of this study are summarized in the following three points:
\begin{enumerate}
    \item We develop a Bayesian model averaging framework for simultaneous variable selection and variable fusion in high-dimensional linear regression. The proposed formulation uses two spike distributions to represent exact exclusion of covariates and exact fusion of adjacent regression coefficients within a unified spike-and-slab model.

    \item We propose a nonlocal slab distribution tailored to simultaneous variable selection and variable fusion. The proposed slab is designed to improve the convergence rate of posterior probabilities for redundant models by assigning low density not only near zero regression coefficients but also near equality between adjacent coefficients that should not be fused. 

    \item We construct a prior distribution over the model space through a ternary latent indicator vector and a Markov-chain-based prior. With suitable initial and transition probabilities,  this construction can induce a uniform prior over the admissible model space. We also provide a posterior sampling scheme for the proposed Bayesian model.
\end{enumerate}

This paper is organized as follows. In Section~\ref{section:preliminaries},  we introduce the model setup,  along with spike-and-slab priors and several nonlocal priors in terms of variable selection.
Section~\ref{sec:Proposed_Method} proposes a Bayesian model averaging approach that simultaneously achieves variable selection and variable fusion,  and presents one prior distribution on the model space that arises naturally under this framework. In addition,  we introduce a new nonlocal prior that is effective for variable fusion and use it as the slab distribution. Section~\ref{sec:Theoretical_Properties} clarifies the asymptotic properties of the proposed method from the viewpoints of model selection and regression coefficient estimation. Section~\ref{sec:Computational_Algorithm} describes a concrete sampling algorithm for implementing the proposed method. Section~\ref{sec:numerical} examines the finite-sample performance of several variable fusion methods through simulations and with spectral data.


\section{Preliminaries}
\label{section:preliminaries}
In this section,  we summarize the model setup,  notation,  and previous studies used throughout this paper. First,  Section~\ref{subsec:model} describes the linear regression model and the basic assumptions. Next,  as related work motivating the present study,  Section~\ref{subsec:spike-and-slab} reviews spike-and-slab priors for variable selection,  and Section~\ref{subsec:NLPs} briefly overviews non-local priors for variable selection.

\subsection{Model Setting and Notation}
\label{subsec:model}
Let $\bm{y} = (y_1,  \ldots,  y_n)^\top$ be an $n$-dimensional real-valued response vector,  and let $\bm{X} = (\bm{x}_1,  \ldots,  \bm{x}_p)$ be an $n \times p$ real-valued design matrix. 
Here $\bm x_j =(x_{1j},  \ldots,  x_{nj})^\top \ (j=1, \ldots, p)$ represents an $n$-dimensional real-valued covariate. 
Without loss of generality,  we can assume that the response vector and covariates satisfy
\begin{align}
    \sum_{i=1}^{n} y_i = 0, 
    \qquad
    \sum_{i=1}^{n} x_{ij} = 0, 
    \qquad
    \sum_{i=1}^{n} x_{ij}^2 = n
    \quad (j = 1,  \ldots,  p).
\end{align}
We assume that the response and covariates satisfy the following linear relationship:
\begin{align}
    \bm{y} = \bm{X}\bm{\theta} + \bm{\varepsilon}.
    \label{eq:LinearRegression}
\end{align}
Here,  $\bm{\theta} = (\theta_1, \ldots,  \theta_p)^\top$ is a $p$-dimensional regression coefficient vector,  and
$\bm{\varepsilon} = (\varepsilon_1,  \ldots,  \varepsilon_n)^\top$ is an $n$-dimensional error vector that follows
$\mathrm{N}(\bm{0}_n,  \sigma^2 \bm{I}_n)$. Here,  $\sigma^2$ is the error variance.
Also,  $\bm{0}_n$ and $\bm{I}_n$ denote the $n$-dimensional zero vector and the $n$-dimensional identity matrix,  respectively. Under these assumptions,  the likelihood function is given by $\mathrm{N}_{n}(\bm{y} \mid \bm{X}\bm{\theta},  \sigma^2 \bm{I}_n)$.

\subsection{BMA and Spike-and-Slab Priors}
\label{subsec:spike-and-slab}
Bayesian model averaging (BMA) is a method that combines multiple models according to their posterior probabilities. For candidate models $M_k \, (k=1, \ldots, K)$,  the posterior mean under BMA is given by
\begin{align}
\E(\bm{\theta} \mid \bm{y})
=
\sum_{k=1}^{K}
\E(\bm{\theta} \mid M_k, \bm{y})P(M_k\mid \bm{y}).
\end{align}
That is,  BMA performs inference by averaging model-specific posterior means with weights given by the posterior probabilities of the models. In practice,  however,  it is often computationally infeasible to evaluate the posterior probabilities for all candidate models,  because the model space is too large. Therefore,  stochastic search methods such as simplified shotgun stochastic search with screening are used to efficiently explore models with relatively large posterior probabilities \citep{Shin2018}.

Some spike-and-slab formulations for variable selection \citep{george1993variable} can be interpreted as BMA. In this framework,  each model is represented by a latent indicator vector,  and posterior computation is commonly carried out within the framework of Gibbs sampling. For example,  in the context of variable selection,  the latent indicator vector $\bm{\delta} \in \{-1, 0\}^p$ is defined as follows:
\begin{align}
    \delta_j =
    \begin{cases}
        0 & \text{if } \theta_j = 0, \\
        -1 & \text{otherwise}
    \end{cases}
    \quad (j = 1, \ldots, p).
\end{align}

In spike-and-slab priors,  different prior distributions are assigned to the regression coefficient $\theta_j$ depending on the value of the latent variable $\delta_j \in \{0, -1\}$. When $\delta_j = 0$,  a prior distribution that strongly expresses the belief that the regression coefficient $\theta_j$ is equal to zero or lies in its neighborhood is assigned; this is called the spike distribution. Specifically,  the examples include a Dirac distribution satisfying $\theta_j = 0$ a.s.~and a normal distribution with small variance. On the other hand,  when $\delta_j = -1$,  a continuous prior distribution that assigns probability over a wide range of values of the regression coefficient $\theta_j$ is assigned; this is called the slab distribution. Various suitable distributions have been proposed for the slab distribution,  including the normal,  Laplace,  and $t$ distributions   \citep{george1993variable, rockova2018spike, ishwaran2005spike}.
In particular,  \citet{Shi2019MassNonlocal} proposed the use of nonlocal priors \citep{johnson2010use} as the slab distribution. The specific forms and properties of {nonlocal priors} in variable selection are summarized in the next subsection.

\subsection{Nonlocal Priors for Variable Selection}
\label{subsec:NLPs}
A nonlocal prior for variable selection refers to a distribution such that the probability density function converges to zero when $\theta_j \to 0$ for any $j=1, \ldots, p$. 
A distribution that does not satisfy this condition is called a local prior.
\citet{rossell2013hdclassifiers} defined the product moment prior (pMOM),  product inverse moment prior (piMOM),  and product exponential moment prior (peMOM) as representative nonlocal priors effective for variable selection. 
These nonlocal priors can be written,  respectively,  as follows:
\begin{align}
\pi_{\mathrm{pMOM}}(\bm{\theta}\mid\sigma^2)
&\propto
\left\{\prod_{i=1}^{p}\theta_{i}^{2}\right\}
\phi\!\left(\bm{\theta}; \bm{0},  \sigma^2 \bm{I}_{p}\right), \\
\pi_{\mathrm{piMOM}}(\bm{\theta}\mid\sigma^2)
&\propto
\prod_{i=1}^{p} |\theta_{i}|^{-2}
\exp\!\left(-\frac{\sigma^2}{\theta_{i}^2}\right), \\
\pi_{\mathrm{peMOM}}(\bm{\theta}\mid\sigma^2)
&\propto
\left\{\prod_{i=1}^{p}\exp\!\left(-\frac{\sigma^2}{\theta_{i}^2}\right)\right\}
\phi\!\left(\bm{\theta}; \bm{0},  \sigma^2 \bm{I}_{p}\right), 
\end{align}
where $\phi(\cdot;\bm{0}, \sigma^2\bm{I}_{p})$ denotes the probability density function of the multivariate normal distribution with mean zeros and covariance $\sigma^2 \bm I_p$. 
Figure~\ref{fig:nonlocal_densities} compares the shapes of the pMOM,  piMOM,  and peMOM priors with the standard normal density under $\sigma^2=1$. 
For each of the nonlocal priors,  unlike the normal distribution,  the probability density approaches zero as the regression coefficient approaches zero.

\begin{figure}[t]
  \centering
  \includegraphics[width=1.00\textwidth]{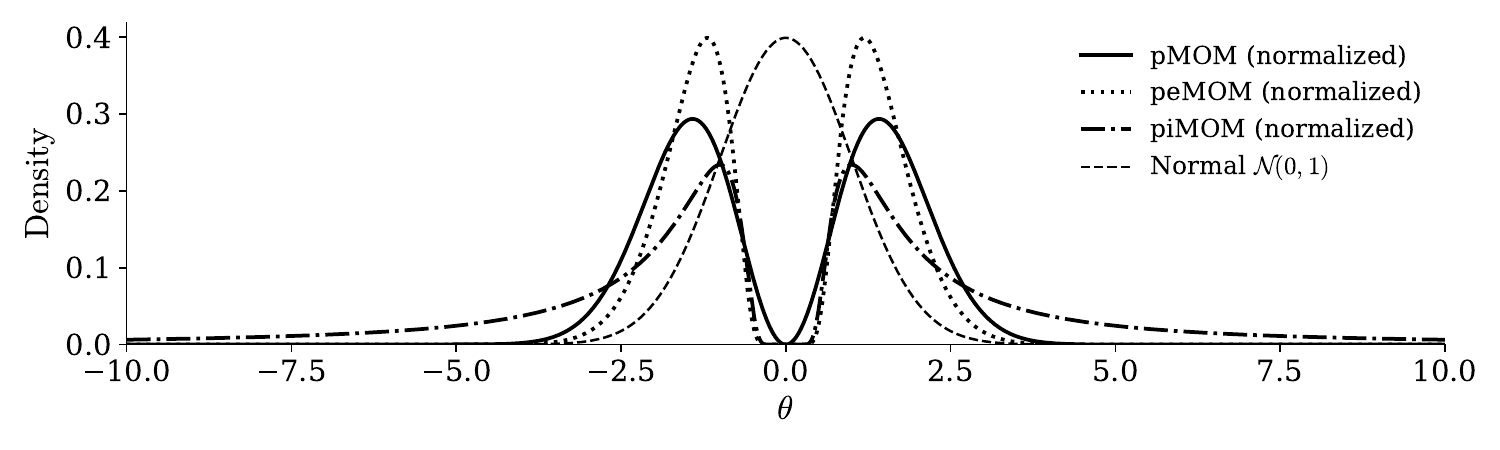} \vspace{-10mm}
  \caption{A comparison of the nonlocal priors (pMOM,  piMOM,  peMOM) with the standard normal density under $\sigma^2=1$.}
  \label{fig:nonlocal_densities}
\end{figure}

It is known that,  regardless of whether local priors or nonlocal priors are used,  for a model $M$ that excludes some or all covariates contributing to the response,  the posterior probability $P(M\mid \bm{y})$ decays at an exponential order under some regularity conditions when the prior probability $P(M)$ and $p$ are fixed \citep{rossell2017nonlocal,  dawid1999trouble}. 
On the other hand,  for a redundant model $M$ that includes all covariates contributing to the response and $r \ (\geq 1)$ additional covariates not contributing,  nonlocal priors have desirable properties.
Under some regularity conditions,  when $P(M)$ and $p$ are fixed,  the posterior probability based on the pMOM prior satisfies $P(M\mid \bm{y})=O_p(n^{-\frac{3}{2}r})$ while that based on a local prior does $P(M\mid \bm{y})=O_p(n^{-\frac{1}{2}r})$ \citep{rossell2017nonlocal}. 
In small sample sizes,  this difference leads to a substantial impact on model estimation accuracy. 


Motivated by the desirable properties exhibited by the pMOM prior,  one of the main objectives of this study is to extend the pMOM prior in a way that holds these properties,  and to construct a framework that can naturally induce variable fusion.

\section{Proposed Method}
\label{sec:Proposed_Method}
In this section,  we propose a Bayesian method for variable selection and variable fusion that adopts a nonlocal prior distribution as the slab distribution. 
Note that the variable fusion considered in this section is allowed only between adjacent variables in the model that includes all covariates. 
For example,  we consider the situation that the $j$th variable is excluded by variable selection. 
Then,  the $(j-1)$th and $(j+1)$th variables are adjacent,  but they are not candidates for fusion because they are not adjacent in the original order. 



\subsection{Latent Variable Vector and its Prior Distribution}

To construct a Bayesian model using the spike-and-slab approach,  we first introduce the following $p$-dimensional indicator vector $\bm{\delta}=(\delta_1, \ldots, \delta_p)^\top$,  whose components take values in $\{-1, 0, 1\}$:
\begin{align}
\delta_1 &=
\begin{cases}
0 & \text{if } \theta_1 = 0, \\
-1 & \text{otherwise}, 
\end{cases}
\quad
\delta_j =
\begin{cases}
1 & \text{if } \theta_j = \theta_{j-1} \neq 0, \\
0 & \text{if } \theta_j = 0, \\
-1 & \text{otherwise}, 
\end{cases}
\quad (j=2, \ldots, p).
\label{eq:def_of_delta}
\end{align}
In other words,  $\delta_j = 1$ for $j=2, \ldots, p$ indicates that the $(j-1)$-th and $j$-th regression coefficients are fused,  while $\delta_j = 0$ for $j=1, \ldots, p$ indicates that the $j$-th regression coefficient is exactly zero.

In spike-and-slab methods,  a prior distribution on the model space is induced by placing a prior distribution on the latent variable vector. For example,  when the objective is only variable selection,  it is common to impose independent and identically distributed Bernoulli distributions on the components of the latent variables. This corresponds to setting the prior probability that each variable is selected to $1/2$. For the three-valued latent variable vector used in the proposed method,  which takes variable fusion into account,  one may consider assuming an independent categorical prior distribution on each component as a natural extension of the Bernoulli distribution:
\begin{align}
    \delta_j \iid \mathrm{Categorical}(a_{-1},  a_0,  a_1),  \quad j = 1, \ldots,  p.
    \label{eq:categorical}
\end{align}
If $a_{-1}=a_0=a_1=1/3$,  this means that a prior probability of $1/3$ is assigned to each of the following cases,  respectively: the case where the variable is fused with the previous regression coefficient,  the case where it is set to zero,  and the case where it is neither of these.

However,  this seemingly natural prior distribution has the drawback that it cannot assign equal prior probabilities to all models. 
For example,  consider the two-dimensional case: $y=\theta_1x_1+\theta_2x_2+\varepsilon$. 
Then,  the full model $y=\theta_1x_1+\theta_2x_2+\varepsilon$ is represented by $\delta=(-1, -1)$. 
Meanwhile,  the null model containing no covariates is represented by $\delta=(0, 0)$ and $\delta=(0, 1)$,  while the model with one coefficient is represented by $\delta=(-1, 0)$,  $\delta=(-1, 1)$,  and $\delta=(0, -1)$. 
If we set $a_{-1}=a_0=a_1=1/3$ for the prior \eqref{eq:categorical},  the probability of the full model is 1/6,  that of the null model is 2/6,  that of the remaining model is each 1/6. 
This means that the distribution \eqref{eq:categorical} with $a_{-1}=a_0=a_1=1/3$ does not provide equal probabilities on all models. 
Furthermore,  it can be shown that no choice of $a_{-1}, a_0, a_1$ allows the distribution \eqref{eq:categorical} to assign equal probabilities to all models.
In general,  the following proposition holds. 
\begin{proposition}
\label{prop:unavailable}
    Consider the indicator vector $\bm{\delta}$ in~(\ref{eq:def_of_delta}) for $p\ge 2$.
If $\delta_1, \ldots, \delta_p$ are mutually independent,  no prior distribution on these parameters can assign equal prior probabilities to all models.
\end{proposition}
\noindent The proof is given in the Supplementary Material \citep{miyake2026supplement}. 
This proposition states that we should not assume an i.i.d. prior distribution for $\delta_1, \ldots, \delta_p$ if we assign equal prior probabilities to all models. 

To assign equal prior probabilities to all assumed models,  we consider a following Markov chain type prior distribution on $\bm{\delta}$:
\begin{align}
\delta_1 &\sim \mathrm{Categorical}(\pi_{-1}, \pi_{0}), \\
\delta_j\mid \delta_{j-1} &\sim \mathrm{Categorical}\bigl(K_j(\delta_{j-1}, -1), \, K_j(\delta_{j-1}, 0), \, K_j(\delta_{j-1}, 1)\bigr), 
\qquad j=2, \ldots, p, 
\end{align}
where $K_j(a, b) = \Pr(\delta_j=b\mid \delta_{j-1}=a)$ represents the transition probability from state $\delta_{j-1}=a$ to state $\delta_{j}=b$. 
Since the indicator $\delta_j$ for $j=2, \ldots, p$ takes values in $\{-1, 0,1\}$,  we can obtain the $3\times 3$ transition matrix in the form
\begin{align}
K_j
& =
\left(
\begin{array}{ccc}
K_j(-1, -1) & K_j(-1, 0) & K_j(-1, 1)\\
K_j(0, -1)  & K_j(0, 0)  & K_j(0, 1)\\
K_j(1, -1)  & K_j(1, 0)  & K_j(1, 1)
\end{array}
\right),  \quad j=2, \ldots, p.
\label{eq:markov_prior_Kdef}
\end{align}
In this study,  the transition matrix is set as
\begin{align}
K_j
&=
\left(
\begin{array}{ccc}
\omega_{j, -1} & \omega_{j, 0} & \omega_{j, 1}\\
\kappa_{j, -1} & \kappa_{j, 0} & 0\\
\omega_{j, -1} & \omega_{j, 0} & \omega_{j, 1}
\end{array}
\right).
\label{eq:markov_prior_main_matrix}
\end{align}
Since we set the $(2, 3)$ component of the transition matrix in Equation \eqref{eq:markov_prior_main_matrix} to zero,  the transition $K_j(0, 1) = \Pr (\delta_j = 1 \mid \delta_{j-1} = 0) \ (j=2, \ldots, p)$ is not allowed. 
By imposing this restriction on the transitions,  a one-to-one correspondence between $\bm{\delta}$ that can transition and the model space is obtained.

Furthermore,  the initial probabilities $(\pi_{-1}, \pi_0)$ and the transition probabilities \\$\{\omega_{j, -1}, \omega_{j, 0}, \omega_{j, 1}, \kappa_{j, -1}, \kappa_{j, 0}\}_{j=2}^p$ can be chosen so that the same prior probability is assigned to all admissible models.
The probabilities are provided according to the following theorem,  and the proof is given in the Supplementary Material \citep{miyake2026supplement}.
\begin{proposition}[Initial probabilities and transition matrices giving equal prior probabilities]
\label{prop:uniform_objective_prior}
Consider the chain prior distribution in \eqref{eq:markov_prior_main_matrix}. For $j=2, \ldots, p$,  let $r_j = p-j+1$,  and let $(F_n)_{n=1}^{\infty}$ be the Fibonacci sequence. 
Set the transition probabilities as follows:
\begin{align}
\omega_{j, -1}=\frac{F_{2r_j}}{F_{2r_j+2}}, \quad
\omega_{j, 0}=\frac{F_{2r_j-1}}{F_{2r_j+2}}, \quad
\omega_{j, 1}=\frac{F_{2r_j}}{F_{2r_j+2}}, \quad
\kappa_{j, -1}=\frac{F_{2r_j}}{F_{2r_j+1}}, \quad
\kappa_{j, 0}=\frac{F_{2r_j-1}}{F_{2r_j+1}}.
\label{eq:omega_kappa_closed_app}
\end{align}
Also,  set the initial probabilities as follows:
\begin{align}
\pi_{-1}=\frac{F_{2p}}{F_{2p+1}}, \quad
\pi_{0}=\frac{F_{2p-1}}{F_{2p+1}}.
\label{eq:pi_closed_app}
\end{align}
Then,  the induced prior distribution assigns equal prior probability to all models that account for variable selection and variable fusion.
\end{proposition}

\subsection{Spike Prior and Fusion-Inducing Slab Prior}

For the spike prior,  we use Dirac measures:
\begin{align}
  \theta_j \mid \delta_j = 0 \quad (j : \delta_j = 0), \\
  \theta_j - \theta_{j-1} \mid \delta_j = 0 \quad (j : \delta_j = 1).
\end{align}
These distributions represent spike distributions for variable selection and variable fusion,  respectively.

For the slab distribution,  we use the following distribution:
\begin{align}
p(\bm{\theta}_{\delta}\mid\bm{\delta}, \sigma^2) = C_p 
\left\{\prod_{j=1}^{p_1}\theta_{\delta_j}^2\right\}
\left\{\prod_{j\in\Lambda}(\theta_{\delta_j}-\theta_{\delta_{j-1}})^2\right\}
(2\pi)^{-p_1/2}\sigma^{-3p_1-2\lvert \Lambda\rvert}
\exp\left(-\frac{1}{2\sigma^2}\bm{\theta}_\delta^\top\bm{\theta}_\delta\right).
\label{eq:fusion-pMOM-prior}
\end{align}
Here, $\bm{\theta}_{\bm{\delta}}$ denotes the vector obtained by extracting only those components of $\bm{\theta}$ for which $\delta_j=-1$, $\Lambda$ is an index set such that $\theta_{\delta_j}$ and $\theta_{\delta_{j-1}}$ can be fused, and $C_p$ is a normalizing constant. When $\Lambda=\{1,\ldots,p_1\}$, $C_p$ is given by
\begin{align}
C_p = \frac{2\sqrt{15}}{(3+\sqrt{15})^{p}-(3-\sqrt{15})^{p}}.
\label{eq:normalizing_constant}
\end{align}
The derivation of this constant is given in the Supplementary Material \citep{miyake2026supplement}. 
For a general $\Lambda$, the corresponding normalizing constant can be derived in an analogous manner.
We call this prior distribution \eqref{eq:fusion-pMOM-prior} \textit{fusion-pMOM prior}.
Similar to the pMOM prior,  the fusion-pMOM prior converges to zero as $\theta_j \to 0$ for any $j=1, \ldots, p$.
Furthermore,  as a different property,  it also converges to zero as $\theta_j - \theta_{j-1} \to 0$ for any $j=2, \ldots, p$.
Figure~\ref{fig:fusion-pmom-prior} illustrates these two properties for the two-dimensional case.

\begin{figure}[t]
  \centering
  \includegraphics[width=0.52\linewidth]{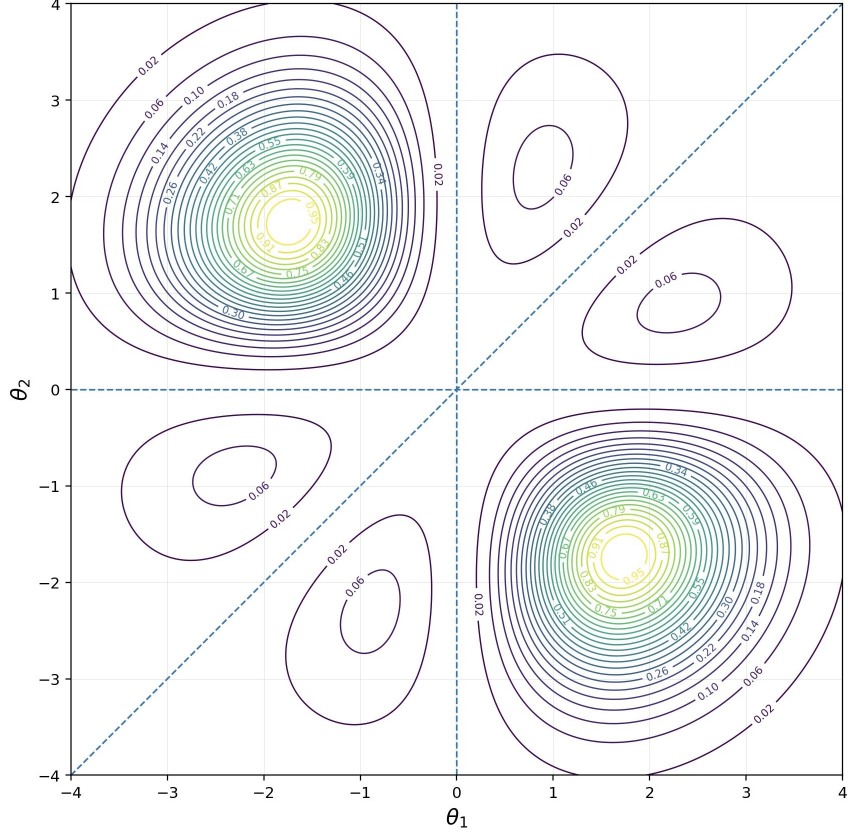}
  \caption{Contour plot of the fusion-pMOM prior for two regression coefficients. The prior density converges to zero as $\theta_1 \to 0$, $\theta_2 \to 0$, and $\theta_2-\theta_1 \to 0$.}
  \label{fig:fusion-pmom-prior}
\end{figure}

In Section~\ref{sec:Theoretical_Properties},  we show that using the fusion-pMOM prior as the slab distribution yields desirable properties for both model selection and regression coefficient estimation in terms of convergence rates.

The parameters $\bm{\omega}$ and $\kappa$ are specified as follows:
\begin{align}
    \bm{\omega}_j^{(k+1)}&\sim\mathrm{Dirichlet}\bigl(Aa_{j, -1}, Aa_{j, 0}, Aa_{j, 1}\bigr), \\
    \kappa_{j, -1}^{(k+1)}&\sim\mathrm{Beta}\bigl(Bc_{j, -1}, Bd_{j, 0}\bigr).
\end{align}
Here,  the hyperparameters are defined in the same manner as in Proposition~\ref{prop:uniform_objective_prior}:
\begin{align}
a_{j, -1}=\frac{F_{2r_j}}{F_{2r_j+2}}, \quad
a_{j, 0}=\frac{F_{2r_j-1}}{F_{2r_j+2}}, \quad
a_{j, 1}=\frac{F_{2r_j}}{F_{2r_j+2}}, \quad
c_{j, -1}=\frac{F_{2r_j}}{F_{2r_j+1}}, \quad
d_{j, 0}=\frac{F_{2r_j-1}}{F_{2r_j+1}}.
\end{align}
Note that the equiprobability of models is preserved even under this hierarchical specification.

\label{subsec:Indicator_vector}

\section{Theoretical Properties}
\label{sec:Theoretical_Properties}

In this section,  we present the theoretical properties of the proposed method. 
We show that the fusion-pMOM prior \eqref{eq:fusion-pMOM-prior},  used as the slab distribution,  suppresses the posterior probability assigned to redundant models more strongly than local priors do. As a result,  the posterior mean attains a favorable convergence rate.
Note that the proofs of propositions and theorems in this section are given in the Supplementary Material  \citep{miyake2026supplement}.

When the number of covariates is $p$,  let $K$ denote the total number of models obtained by taking both variable selection and variable fusion into account. Here,  $K$ coincides with the $(2p+1)$-st term of the Fibonacci sequence.
In what follows,  we index the models in the model space by $k=1, \ldots, K$ and denote them by $\bm{\delta}_k$. 
Their dimensions,  corresponding sampling densities for $\bm{y}$,  and marginal likelihoods are denoted by $p_k$,  $f_k(\bm{y}\mid \bm{\theta}_k, \phi_k)$,  and $m_k(\bm{y})$,  respectively. These models are represented by the latent variable vector defined in Equation (\ref{eq:def_of_delta}). 
Among these models,  let $\bm{\delta}_t$ denote the true model that generates the data $\bm{y}$. If $\bm{\delta}_k$ is either identical to the true model $\bm{\delta}_t$ or a supermodel containing it,  we write $\bm{\delta}_t \subset \bm{\delta}_k$,  and otherwise we write $\bm{\delta}_t \not\subset \bm{\delta}_k$. 
In addition,  we assume that $p$ is fixed. 

We impose the Walker condition \citep{walker1969asymptotic} and the following two conditions whenever needed:
\begin{itemize}
    \item[A.] Suppose that there exists a unique true model $\bm{\delta}_t$ with the smallest dimension $p_t$. 
Assume that,  for any $k$ satisfying $\bm{\delta}_t \not\subset \bm{\delta}_k$, 
\begin{align*}
\mathrm{KL}\bigl(f_t(y_n\mid \bm{\theta}_t^*, \phi_t^*), \, f_k(y_n\mid \bm{\theta}_k, \phi_k)\bigr) > 0
\end{align*}
holds for all $(\bm{\theta}_k, \phi_k)$.

    \item[B.] In Condition A,  suppose that $\phi^*$ and $\theta_i^*$ are fixed. 
\end{itemize}

We first investigate the effect of introducing a non-local prior on model selection through the convergence rate of the Bayes factor. 
The following proposition shows that,  for redundant models containing the true model,  the Bayes factor decays faster than in the case of local priors.

\begin{proposition}[Convergence rate of the Bayes factor]\label{thm:BF_convergence}
Assume that the model $\bm{\delta}_k$ satisfies the Walker condition and Condition A,  and define
\begin{align}
\mathrm{BF}_{kt}:=\frac{m_k(\bm{y})}{m_t(\bm{y})}
\end{align}
as the Bayes factor relative to $\bm{\delta}_t$.
Then,  the following holds:
\begin{align}
\bm{\delta}_t \subset \bm{\delta}_k
&\Longrightarrow
\mathrm{BF}_{kt}
=
O_p\!\left(n^{-\frac{3}{2}(p_k-p_t)}\right), 
\label{eq:BF_overfit}
\\
\bm{\delta}_t \not\subset \bm{\delta}_k
&\Longrightarrow
\mathrm{BF}_{kt}
=
O_p(e^{-Cn})
\quad \text{for some } C>0 .
\label{eq:BF_underfit}
\end{align}
\end{proposition}

In the situation that $\bm{\delta}_t \subset \bm{\delta}_k$,  while local priors yield the decay rate $O_p(n^{-(p_k-p_t)/2})$ \citep{dawid1999trouble},  the proposed method yields the faster rate $O_p(n^{-3(p_k-p_t)/2})$. 
On the other hand,  in the situation that $\bm{\delta}_t \not\subset \bm{\delta}_k$,  the Bayes factor decays at an exponential rate,  as in the case of local priors. Therefore,  the primary effect of the proposed method is to suppress the posterior probability of redundant models in small samples.

While Proposition~\ref{thm:BF_convergence} shows an improved convergence rate of the Bayes factor for redundant models,  it is not clear from this result alone how this improvement is reflected in the asymptotic behavior of the posterior mean for the regression coefficients. To clarify this,  it is first necessary to investigate the asymptotic behavior of the posterior mean for the regression coefficients,  $\E(\theta_j \mid \bm{y},  \bm{\delta}_k)$,  under each fixed model. 
The following proposition provides the corresponding explicit convergence rates.

\begin{proposition}[Conditional posterior mean under local rate assumptions]
\label{prop:sparse_component_mean_local}
Assume that $P(\bm{y}\mid \bm{\theta}, \phi)$ satisfies the Walker condition and that Conditions A and B hold.
Define
$\Delta_i^*
=
\theta_i^*-\theta_{i-1}^*$.
Then the following assertions hold:
\begin{align}
&\theta_i^*\neq 0
\ \Longrightarrow\
\E(\theta_i\mid \bm{y})
=
\theta_i^*+O_p(n^{-1/2}),
\label{eq:mean_rate_fixed_final}
\\
&\theta_i^*=0,\quad
\Delta_i^*\neq 0,\quad
\Delta_{i+1}^*\neq 0
\ \Longrightarrow\
\E(\theta_i\mid \bm{y})
=
O_p(n^{-1/2}),
\label{eq:mean_rate_zero_isolated_final}
\\
&\theta_i^*=0,\quad
\text{exactly one of $\Delta_i^*$ and $\Delta_{i+1}^*$ is equal to zero}
\ \Longrightarrow\
\E(\theta_i\mid \bm{y})
=
O_p(n^{-1/3}),
\label{eq:mean_rate_zero_onefusion_final}
\\
&\theta_i^*=0,\quad
\Delta_i^*=0,\quad
\Delta_{i+1}^*=0
\ \Longrightarrow\
\E(\theta_i\mid \bm{y})
=
O_p(n^{-1/4}).
\label{eq:mean_rate_zero_twofusion_final}
\end{align}
\end{proposition}


Proposition~\ref{prop:sparse_component_mean_local} shows that,  under a fixed model,  the conditional posterior mean for the regression coefficient has convergence rate $\E(\theta_j \mid \bm{y}, \bm{\delta})=O_p(n^{-1/2})$. 
That is,  even under a fixed model,  the conditional posterior mean for the regression coefficient has a convergence rate comparable to that under local priors.

Finally,  in order to investigate the effectiveness of the proposed method,  we clarify how model averaging improves the accuracy of regression coefficient estimation.

\begin{theorem}[Convergence rate of the posterior mean]
\label{prop:expectation}
Let $\bm{\delta}_t$ denote the true model. Assume that the Walker condition and Conditions A and B hold.
Define
$\Delta_i^*
=
\theta_i^*-\theta_{i-1}^*$.
Then, for the posterior mean
\begin{align}
\E(\theta_i\mid \bm y_n)
=
\sum_k
\E(\theta_i\mid \bm{\delta}_k,\bm y_n)
P(\bm{\delta}_k\mid \bm y_n),
\end{align}
the following holds:
\begin{align}
&\theta_i^*\neq 0
\Longrightarrow
\E(\theta_i\mid \bm y_n)
=
\theta_i^*+O_p(n^{-1/2}),
\\
&\theta_i^*=0,\quad
\Delta_i^*\neq 0,\quad
\Delta_{i+1}^*\neq 0
\Longrightarrow
\E(\theta_i\mid \bm y_n)
=
O_p(n^{-2}),
\\
&\theta_i^*=0,\quad
\text{exactly one of $\Delta_i^*$ and $\Delta_{i+1}^*$ is equal to zero}
\Longrightarrow
\E(\theta_i\mid \bm y_n)
=
O_p(n^{-11/6}),
\\
&\theta_i^*=0,\quad
\Delta_i^*=0,\quad
\Delta_{i+1}^*=0
\Longrightarrow
\E(\theta_i\mid \bm y_n)
=
O_p(n^{-7/4}).
\end{align}
\end{theorem}

\begin{theorem}[Convergence rate of the posterior mean for adjacent differences]
\label{thm:posterior_mean_difference_fixed_dim}
Let $\bm{\delta}_t$ denote the true model. Assume that the Walker condition and Conditions A and B hold.
For each model $\bm{\delta}_k$,  define
\begin{align}
\E(\Delta_j\mid \bm{y})
=
\sum_k \E(\Delta_j\mid \bm{\delta}_k, \bm{y})\, P(\bm{\delta}_k\mid \bm{y}), 
\label{eq:bma_mean_difference}
\end{align}
where $\Delta_j:=\theta_j-\theta_{j-1}$ and
$\Delta_j^*:=\theta_j^*-\theta_{j-1}^*$.
Then the posterior mean satisfies
\begin{align}
&\Delta_j^*\neq0
\Longrightarrow
\E(\Delta_j\mid \bm{y})
=
\Delta_j^*+O_p(n^{-1/2}), 
\label{eq:thm_diff_fixed_nonzero}
\\
&\Delta_j^*=0, \
\Delta_{j-1}^*\neq 0, \
\Delta_{j+1}^*\neq 0
\Longrightarrow
\E(\Delta_j\mid \bm{y})
=
O_p(n^{-2}), 
\label{eq:thm_diff_zero_isolated}
\\
&\Delta_j^*=0, \
\text{exactly one of }\Delta_{j-1}^* \text{ and } \Delta_{j+1}^*
\text{ is equal to zero}
\Longrightarrow
\E(\Delta_j\mid \bm{y})
=
O_p(n^{-11/6}), 
\label{eq:thm_diff_zero_onefusion}
\\
&\Delta_j^*=0, \
\Delta_{j-1}^*=0, \
\Delta_{j+1}^*=0
\Longrightarrow
\E(\Delta_j\mid \bm{y})
=
O_p(n^{-7/4}).
\label{eq:thm_diff_zero_twofusion}
\end{align}
\end{theorem}

Under local priors,  it is known that the convergence rate remains $O_p(n^{-1/2})$ in all of these settings \citep{dawid1999trouble}. 
In contrast,  Theorems~\ref{prop:expectation} and \ref{thm:posterior_mean_difference_fixed_dim} show that,  under the proposed method,  faster convergence rates are obtained for both unnecessary regression coefficients and unnecessary adjacent differences. 
This implies that the proposed method has the property of inducing stronger shrinkage on unnecessary regression coefficients and unnecessary adjacent differences.

On the other hand, under the pMOM prior for variable selection, regression coefficients whose true values are zero uniformly attain the convergence rate $O_p(n^{-2})$.
In contrast, under the proposed method, when $\theta_i^*=0$, the upper bound on the convergence rate is suggested to be of the same order or slower because of the structure that promotes variable fusion.
However, the $O_p$ evaluations established here provide only upper bounds on the convergence rates and do not imply that the actual convergence rates coincide with these orders.
Therefore, these upper bounds alone do not allow us to conclude that the introduction of variable fusion actually reduces the estimation accuracy of the regression coefficients.
Nevertheless, the possibility that a trade-off arises between variable selection and variable fusion in estimating regression coefficients whose true values are zero cannot be ruled out.

At the same time,  however,  as shown in Theorem~\ref{thm:posterior_mean_difference_fixed_dim},  the proposed method induces strong shrinkage on the differences between adjacent regression coefficients. 
This is due to a feature of the fusion-pMOM prior,  which promotes not only variable selection but also variable fusion simultaneously. 
Therefore,  the proposed method can be regarded as a method that achieves fast convergence rates for both variable selection and variable fusion in regression coefficients.

\section{Computational Algorithm}
\label{sec:Computational_Algorithm}
The sampling under the proposed method is carried out by Gibbs sampling. In Section~\ref{subsec:gibbs_updates_all},  we provide the Gibbs sampling algorithm. In Section~\ref{subsec:no-local_sampling},  we describe a sampling method from the non-local prior distribution.

\subsection{Gibbs Updates}
\label{subsec:gibbs_updates_all}

For notational convenience,  we define
\begin{align*}
Q_{\bm{\delta}}(\bm{\theta}_{\bm{\delta}})
&:=
\prod_{j=1}^{p_{\bm{\delta}}}\theta_{\bm{\delta}_j}^2
\prod_{j\in\Lambda_{\bm{\delta}}}(\theta_{\bm{\delta}_j}-\theta_{\bm{\delta}_{j-1}})^2, 
&
A_{\bm{\delta}}
&:=
X_{\bm{\delta}}^{\top}X_{\bm{\delta}}+I,  \\
\tilde{\bm{\beta}}_{\bm{\delta}}
&:=
A_{\bm{\delta}}^{-1}X_{\bm{\delta}}^{\top}\bm{y}, 
&
R_{\bm{\delta}}
&:=
\bm{y}^{\top}(I-X_{\bm{\delta}}A_{\bm{\delta}}^{-1}X_{\bm{\delta}}^{\top})\bm{y},  \\
\nu_{\bm{\delta}}
&:=
n+2p_{\bm{\delta}}+2|\Lambda_{\bm{\delta}}|+2\alpha, 
&
s_{\bm{\delta}}^{2}
&:=
\frac{R_{\bm{\delta}}+2\psi}{\nu_{\bm{\delta}}}.
\end{align*}
Using these notations,  the Gibbs sampling algorithm is given as follows. Here,  the indicator function is denoted by $\bm{1}_{\{\cdot\}}$.

\begin{enumerate}
\item For each $j=p, \ldots, 1$,  sample\\
$
\begin{aligned}
\delta_j^{(k+1)}
&\sim \mathrm{Categorical}\!\left(h_{j, -1}^{(k)}, \, h_{j, 0}^{(k)}, \, h_{j, 1}^{(k)}\right), 
\end{aligned}
$\\
where, 
$
\begin{aligned}
h_{j, x}^{(k)}
&\coloneqq
\frac{\tilde h_{j, x}^{(k)}}{\sum_{s\in\{-1, 0, 1\}}\tilde h_{j, s}^{(k)}}, 
\qquad x\in\{-1, 0, 1\}, 
\end{aligned}
$
\\[-2pt]
\[
\bm{\delta}^{(k, j\leftarrow x)}
\coloneqq
(\delta_1^{(k)}, \ldots, \delta_{j-1}^{(k)}, x, \delta_{j+1}^{(k)}, \ldots, \delta_p^{(k)})^\top, 
\]
\[
\tilde h_{j, x}^{(k)}
\coloneqq
\begin{cases}
m\!\left(\bm{y}\mid \bm{\delta}^{(k, j\leftarrow x)}\right)\, 
p^{(j)}_{\delta_{j-1}^{(k)}, x}\, 
p^{(j+1)}_{x, \delta_{j+1}^{(k)}}, 
& (\text{if }j=2, \ldots, p-1), \\[6pt]
m\!\left(\bm{y}\mid \bm{\delta}^{(k, 1\leftarrow x)}\right)\, 
\pi_{x}\, 
p^{(2)}_{x, \delta_{2}^{(k)}}, 
& (\text{if }j=1), \\[6pt]
m\!\left(\bm{y}\mid \bm{\delta}^{(k, p\leftarrow x)}\right)\, 
p^{(p)}_{\delta_{p-1}^{(k)}, x}, 
& (\text{if }j=p).
\end{cases}
\]


\item
$
\bm{\theta}_{\bm{\delta}^{(k+1)}}^{(k+1)}
\sim
\mathrm{fusion\text{-}pMOM}
\left(
\tilde{\bm{\beta}}_{\bm{\delta}^{(k+1)}}, 
(\sigma^2)^{(k)}A_{\bm{\delta}^{(k+1)}}^{-1}
\right), 
$\\
where
$
\mathrm{fusion\text{-}pMOM}(\bm{\mu}, \Sigma)
$
denotes the distribution with probability density function proportional to
\begin{align*}
\phi_d(\bm{\theta};\bm{\mu}, \Sigma)
\prod_{j\in S_{\bm{\delta}}}\theta_j^2
\prod_{j\in\Lambda_{\bm{\delta}}}(\theta_j-\theta_{j-1})^2, 
\end{align*}
where $\phi_d(\bm{\theta};\bm{\mu}, \Sigma)$ is the density of the $d$-dimensional normal distribution with mean $\bm{\mu}$ and covariance matrix $\Sigma$.

\item
$
\begin{aligned}
(\sigma^2)^{(k+1)}
&\sim
\mathrm{IG}\!\left(
\frac{
n+3p_{\bm{\delta}^{(k+1)}}+2|\Lambda_{\bm{\delta}^{(k+1)}}|+2\alpha
}{2}, \, 
\psi+\frac{
\bigl\|\bm{y}-X_{\bm{\delta}^{(k+1)}}\bm{\theta}_{\bm{\delta}^{(k+1)}}^{(k+1)}\bigr\|_2^2
+\bigl(\bm{\theta}_{\bm{\delta}^{(k+1)}}^{(k+1)}\bigr)^\top
\bm{\theta}_{\bm{\delta}^{(k+1)}}^{(k+1)}
}{2}
\right).
\end{aligned}
$

\item For each $j=2, \ldots, p$,  sample

\hspace*{1em}if $\delta_{j-1}^{(k+1)}\in\{-1, 1\}$,  then

\hspace*{2em}
$\bm{\omega}_j^{(k+1)}
\sim
\mathrm{Dirichlet}\Bigl(
Aa_{j, -1}+\bm{1}_{\{\delta_j^{(k+1)}=-1\}}, \
Aa_{j, 0}+\bm{1}_{\{\delta_j^{(k+1)}=0\}}, \
Aa_{j, 1}+\bm{1}_{\{\delta_j^{(k+1)}=1\}}
\Bigr).$

\hspace*{1em}if $\delta_{j-1}^{(k+1)}=0$,  then

\hspace*{2em}
$\bm{\omega}_j^{(k+1)}
\sim
\mathrm{Dirichlet}\bigl(
Aa_{j, -1}, Aa_{j, 0}, Aa_{j, 1}
\bigr).$

\item For each $j=1, \ldots, p$,  sample

\hspace*{1em}if $\delta_{j-1}^{(k+1)}=0$,  then

\hspace*{2em}
$\kappa_{j, -1}^{(k+1)}
\sim
\mathrm{Beta}\Bigl(
Bc_{j, -1}+\bm{1}_{\{\delta_j^{(k+1)}=-1\}}, \
Bd_{j, 0}+\bm{1}_{\{\delta_j^{(k+1)}=0\}}
\Bigr).$

\hspace*{1em}if $\delta_{j-1}^{(k+1)}\in\{-1, 1\}$,  then

\hspace*{2em}
$\kappa_{j, -1}^{(k+1)}
\sim
\mathrm{Beta}\bigl(Bc_{j, -1}, Bd_{j, 0}\bigr).$

\end{enumerate}
The marginal likelihood is given by
\begin{align}
m(\bm{y}\mid\bm{\delta})
&=
C_{\bm{\delta}}(2\pi)^{-(n+p_{\bm{\delta}})/2}\frac{\psi^{\alpha}}{\Gamma(\alpha)}
2^{\frac{\nu_{\bm{\delta}}+p_{\bm{\delta}}}{2}}
\pi^{p_{\bm{\delta}}/2}
|A_{\bm{\delta}}|^{-1/2}
(\nu_{\bm{\delta}}s_{\bm{\delta}}^{2})^{-\nu_{\bm{\delta}}/2}
\Gamma\left(\frac{\nu_{\bm{\delta}}}{2}\right)
\E_{t}\left[Q_{\bm{\delta}}(\bm{\theta}_{\bm{\delta}})\right].
\end{align}
Here,  $\E_t[\cdot]$ denotes expectation with respect to the $p_{\bm{\delta}}$-dimensional $t$ distribution with mean $\tilde{\bm{\beta}}_{\bm{\delta}}$,  scale matrix $s_{\bm{\delta}}^{2}A_{\bm{\delta}}^{-1}$,  and degrees of freedom $\nu_{\bm{\delta}}$. The derivation of the marginal likelihood and the Gibbs update formulas are given in the Supplementary Material  \citep{miyake2026supplement}.

\subsection{Sampling from the Non-Local Prior Distribution}
\label{subsec:no-local_sampling}

Sampling from the fusion-pMOM distribution can be reduced to sampling from the following truncated multivariate normal distribution:
\begin{enumerate}
\renewcommand{\labelenumi}{(\arabic{enumi})}

\item
\(
\begin{aligned}
s_j^{(k+1)} &\sim \mathrm{Unif}\!\left(0, (\theta_j^{(k)})^2\right), 
\qquad \lambda_j^{(k+1)} \coloneqq \sqrt{s_j^{(k+1)}}, 
\qquad j=1, \ldots, p, \\
t_j^{(k+1)} &\sim \mathrm{Unif}\!\left(0, (\theta_j^{(k)}-\theta_{j-1}^{(k)})^2\right), 
\qquad \eta_j^{(k+1)} \coloneqq \sqrt{t_j^{(k+1)}}, 
\qquad j \in \Lambda.
\end{aligned}
\)

\item 
\(
\begin{aligned}
\bm{\theta}^{(k+1)}
&\sim
\mathrm{N}_p(\bm{\mu}, \bm{\Sigma})
\prod_{j=1}^{p}\bm{1}_{\{|\theta_j|>\lambda_j^{(k+1)}\}}
\prod_{j\in \Lambda}\bm{1}_{\{|\theta_j-\theta_{j-1}|>\eta_j^{(k+1)}\}}.
\end{aligned}
\)
\end{enumerate}
This can be derived by straightforward algebraic calculation by using Corollary 2 of \citet{rossell2017nonlocal}.

\section{Numerical Experiments}
\label{sec:numerical}
Hereafter, we refer to the proposed method as the Bayesian fused nonlocal prior method (BFN). 
In this section, we validate BFN through numerical experiments. 
In Section~\ref{subsec:montecarlo}, we conduct a Monte Carlo simulation study. 
In Section~\ref{subsec:real_data_analysis}, we apply BFN to spectral data to predict the composition of products.
\subsection{Monte Carlo Simulation}

\label{subsec:montecarlo}
We generate data according to the following true linear regression model:
\begin{align}
    \bm{y} = X \bm{\theta}^* + \bm{\varepsilon}, 
\end{align}
 where $\bm{\theta}^*$ denotes the $p$-dimensional true regression coefficient vector,  $X$ denotes the $n \times p$ design matrix,  and $\bm{\varepsilon}$ denotes an $n$-dimensional error vector following $\mathrm{N}_n(\bm{0}, \sigma^2 I_n)$. In addition,  each row $\bm{x}_j\ (j=1, \ldots, n)$ of the design matrix $X$ is independently generated from the multivariate normal distribution $\mathrm{N}_p(\bm{0}, \Sigma)$ with common covariance matrix $\Sigma$. The $(i, j)$th element of $\Sigma$ is given by
\begin{align}
    \Sigma_{i, j}
    =
    \begin{cases}
        1,  & i=j, \\
        \rho,  & i\neq j
    \end{cases}
\end{align}
for all $i$ and $j$.

In the numerical experiments,  we consider three settings to examine the performance of BFN under different true regression coefficient vectors,  correlations among covariates,  error variances,  and sample sizes.
We consider the following three cases:

\begin{itemize}
    \item[] Case 1: $\bm{\theta}^*=(\bm{3}_5^\top, \bm{5}_5^\top, \bm{3}_5^\top, \bm{5}_5^\top)^\top$, 
$\rho=0,0.5$, 
$\sigma=0.5$, 
and $n=20, 40, 60, 80, 100$.
    \item[] Case 2: 
$\bm{\theta}^*=(\bm{0}_{140}^\top, \bm{-1.5}_5^\top, \bm{1.5}_5^\top)^\top$, 
$\rho=0,0.5$, 
$\sigma=0.5$, 
and $n=50, 70$.
    \item[] Case 3: 
$\bm{\theta}^*=(\bm{0}_{495}^\top, \bm{3}_5^\top)^\top$, 
$\rho=0,0.5$, 
$\sigma=0.5$, 
and $n=100, 200$.
\end{itemize}
Roughly speaking,  Case 1 corresponds to the low-dimensional setting,  while Cases 2 and 3 correspond to the high-dimensional setting.

For comparison, we consider Fused Lasso (FL) \citep{tibshirani2005fusedlasso}, Bayesian Fused Lasso (BFL) \citep{Kyung2010BayesFusedLasso}, Bayesian Fused Lasso via Horseshoe Prior (BFLH) \citep{kakikawa2023bayesian}, and Bayesian Fused Lasso via Spike-and-slab (BFLSS) \citep{wu2021variable}. 
Since BFLSS did not yield stable results in the high-dimensional setting, it is included only in the comparison for Case 1.

For FL, the regularization parameter was selected by using the \texttt{penalized} package in R. 
For the hyperparameters of the Gamma distributions in BFL and BFLH, we set the shape and rate parameters to 1 and 10, respectively, so that the resulting prior distribution would be relatively flat and weakly informative. 
This setting is also recommended in \citet{Kyung2010BayesFusedLasso}. 
On the other hand, for BFN, the hyperparameters were set to $A=B=1, \ \phi=\psi=1$.

BFL, BFLH, BFLSS, and BFN were all implemented by using Gibbs sampling. 
The number of posterior samples was set to 8000, and the first 2000 samples were discarded as burn-in.

To evaluate estimation accuracy,  we use the mean squared error (MSE),  defined by
\begin{align}
    \mathrm{MSE}
    =
    (\hat{\bm{\theta}}-\bm{\theta}^*)^\top(\hat{\bm{\theta}}-\bm{\theta}^*), 
\end{align}
where $\hat{\bm{\theta}}$ denotes an estimator of the regression coefficient vector $\bm{\theta}$. 
We also use the prediction squared error (PSE) to evaluate predictive performance. The PSE is defined by
\begin{align}
    \mathrm{PSE}
    =
    (\hat{\bm{\theta}}-\bm{\theta}^*)^\top
    \Sigma
    (\hat{\bm{\theta}}-\bm{\theta}^*).
\end{align}
Furthermore,  as a measure for evaluating the accuracy of variable selection and variable fusion,  we define $P_B$ by
\begin{align}
    P_B
    &=
    \frac{
        N_{\mathrm{selection}}
        +
        N_{\mathrm{fusion}}
        +
        N_{\mathrm{otherwise}}
    }{p},
\end{align}
where
\begin{align}
    N_{\mathrm{selection}}
    &\coloneqq
    \#\left\{
    j=1,\ldots,p
    \,\middle|\,
    \theta_j^*=0,\ 
    \hat{\theta}_j=0
    \right\},
    \\
    N_{\mathrm{fusion}}
    &\coloneqq
    \#\left\{
    j=2,\ldots,p
    \,\middle|\,
    \theta_j^*\neq 0,\ 
    \theta_j^*-\theta_{j-1}^*=0,\ 
    \hat{\theta}_j\neq 0,\ 
    \hat{\theta}_j-\hat{\theta}_{j-1}=0
    \right\},
    \\
    N_{\mathrm{otherwise}}
    &\coloneqq
    \#\left\{
    j=1,\ldots,p
    \,\middle|\,
    \theta_j^*\neq 0,\ 
    \theta_j^*-\theta_{j-1}^*\neq 0,\ 
    \hat{\theta}_j\neq 0,\ 
    \hat{\theta}_j-\hat{\theta}_{j-1}\neq 0
    \right\}.
\end{align}
That is,  $N_{\mathrm{selection}}$ represents the number of components for which a coefficient that is truly zero is correctly estimated as zero,  while $N_{\mathrm{fusion}}$ represents the number of pairs of nonzero coefficients with a truly zero adjacent difference for which that difference is correctly estimated as zero. Therefore,  $P_B$ is an index that comprehensively evaluates the accuracy of both variable selection and variable fusion.

The results of the numerical experiments are reported in Tables~\ref{tab:case1},  \ref{tab:case2},  and \ref{tab:case3}. 
In Case 1, we first focus on $P_B$. BFN shows favorable values in all settings and, in particular, achieves the best or comparable performance relative to the competing methods. Moreover, the values of $P_B$ approach one as the sample size increases, indicating that BFN can recover the true model more accurately for larger sample sizes. These results suggest that BFN can appropriately identify the true model involving both variable selection and variable fusion.

We next examine MSE. BFN achieves the smallest value in many settings and, in particular, for $\rho=0.5$, it attains the best performance except when $n=60$. This result is consistent with the theoretical property that BFN can recover the true model with high accuracy, thereby leading to stable performance in coefficient estimation.

Finally, with respect to PSE, BFN achieves the best or nearly best performance in many settings. This indicates that the improved performance in recovering the true model and estimating the coefficients does not come at the cost of predictive accuracy. Rather, BFN can identify the true model while maintaining favorable predictive performance. Overall, in Case 1, BFN demonstrates superior performance compared with the competing methods in terms of true model recovery, coefficient estimation accuracy, and predictive accuracy.

In Cases 2 and 3, the same tendency is maintained even in high-dimensional settings. In particular, BFN achieves the best $P_B$ values in all settings, indicating that it can recover the true model with high accuracy even when the number of covariates is large relative to the sample size. In terms of MSE, BFN also attains the smallest or nearly smallest values in many settings. Furthermore, with respect to PSE, it achieves the best performance in all settings of Cases 2 and 3. Overall, these results demonstrate that BFN performs favorably not only in low-dimensional settings but also in high-dimensional settings, in terms of true model recovery, coefficient estimation accuracy, and predictive accuracy.

\begin{table}[htbp]
\centering
\caption{Average MSE,  PSE,  and $P_B$ in Case 1 (over 100 trials)}
\label{tab:case1}
\renewcommand{\arraystretch}{1.15}
\resizebox{\textwidth}{!}{%
\begin{tabular}{c l c c c c c c c c c c c c}
\hline
& & \multicolumn{6}{c}{$\rho = 0$} & \multicolumn{6}{c}{$\rho = 0.5$} \\
\cline{3-8} \cline{9-14}
$n$ & Method
& MSE & (sd) & PSE & (sd) & $P_B$ & (sd)
& MSE & (sd) & PSE & (sd) & $P_B$ & (sd) \\
\hline

\multirow{5}{*}{20}
& FL       & 0.217 & 0.224 & 0.223 & 0.185 & 0.784 & 0.067 & 0.258 & 0.215 & 0.383 & 0.270 & 0.761 & 0.051 \\
& BFL      & 0.195 & 0.163 & 0.242 & 0.178 & 0.682 & 0.058 & 0.259 & 0.227 & \textbf{0.343} & 0.262 & 0.819 & 0.042 \\
& BFLH     & 0.188 & 0.133 & 0.251 & 0.126 & 0.797 & 0.025 & 0.276 & 0.194 & 0.417 & 0.189 & 0.774 & 0.052 \\
& BFLSS    & 0.261 & 0.162 & 0.326 & 0.170 & 0.782 & 0.050 & 0.293 & 0.222 & 0.366 & 0.090 & 0.815 & 0.019 \\
& BFN & \textbf{0.123} & 0.221 & \textbf{0.138} & 0.247 & \textbf{0.882} & 0.041 & \textbf{0.224} & 0.178 & 0.360 & 0.157 & \textbf{0.823} & 0.027 \\
\hline

\multirow{5}{*}{40}
& FL       & 0.238 & 0.265 & 0.223 & 0.305 & 0.647 & 0.068 & 0.251 & 0.258 & 0.390 & 0.235 & 0.659 & 0.060 \\
& BFL      & 0.303 & 0.239 & 0.371 & 0.098 & 0.798 & 0.032 & 0.254 & 0.239 & 0.259 & 0.125 & 0.800 & 0.022 \\
& BFLH     & 0.231 & 0.220 & 0.286 & 0.113 & \textbf{0.896} & 0.027 & 0.276 & 0.231 & 0.269 & 0.115 & 0.884 & 0.020 \\
& BFLSS    & 0.256 & 0.340 & 0.233 & 0.207 & 0.834 & 0.057 & 0.279 & 0.194 & 0.302 & 0.082 & 0.847 & 0.013 \\
& BFN & \textbf{0.130} & 0.106 & \textbf{0.155} & 0.229 & \textbf{0.896} & 0.041 & \textbf{0.212} & 0.245 & \textbf{0.218} & 0.265 & \textbf{0.910} & 0.066 \\
\hline

\multirow{5}{*}{60}
& FL       & 0.247 & 0.238 & 0.278 & 0.188 & 0.727 & 0.023 & 0.249 & 0.258 & 0.277 & 0.148 & 0.720 & 0.024 \\
& BFL      & 0.293 & 0.203 & 0.342 & 0.110 & 0.726 & 0.030 & 0.257 & 0.229 & 0.275 & 0.097 & 0.732 & 0.019 \\
& BFLH     & \textbf{0.205} & 0.191 & 0.235 & 0.198 & 0.895 & 0.056 & \textbf{0.211} & 0.214 & 0.329 & 0.143 & 0.898 & 0.036 \\
& BFLSS    & 0.283 & 0.230 & 0.322 & 0.183 & 0.828 & 0.031 & 0.285 & 0.162 & 0.340 & 0.168 & 0.836 & 0.043 \\
& BFN & 0.229 & 0.250 & \textbf{0.221} & 0.096 & \textbf{0.920} & 0.029 & 0.261 & 0.190 & \textbf{0.210} & 0.089 & \textbf{0.932} & 0.027 \\
\hline

\multirow{5}{*}{80}
& FL       & 0.213 & 0.236 & 0.187 & 0.314 & 0.876 & 0.072 & 0.235 & 0.198 & 0.360 & 0.159 & 0.904 & 0.025 \\
& BFL      & 0.257 & 0.192 & 0.304 & 0.116 & 0.874 & 0.021 & 0.207 & 0.236 & 0.266 & 0.069 & 0.908 & 0.008 \\
& BFLH     & \textbf{0.109} & 0.283 & \textbf{0.115} & 0.283 & 0.871 & 0.056 & 0.141 & 0.199 & \textbf{0.097} & 0.109 & 0.907 & 0.024 \\
& BFLSS    & 0.157 & 0.199 & 0.134 & 0.284 & 0.941 & 0.098 & 0.144 & 0.240 & 0.231 & 0.203 & 0.917 & 0.059 \\
& BFN & 0.123 & 0.109 & 0.184 & 0.224 & \textbf{0.977} & 0.044 & \textbf{0.108} & 0.145 & 0.114 & 0.250 & \textbf{0.943} & 0.056 \\
\hline

\multirow{5}{*}{100}
& FL       & 0.278 & 0.210 & 0.323 & 0.268 & 0.969 & 0.057 & 0.228 & 0.164 & 0.337 & 0.114 & 0.939 & 0.013 \\
& BFL      & 0.283 & 0.121 & 0.338 & 0.205 & 0.915 & 0.041 & 0.208 & 0.154 & 0.167 & 0.241 & 0.952 & 0.078 \\
& BFLH     & 0.117 & 0.147 & 0.162 & 0.226 & 0.964 & 0.052 & 0.119 & 0.180 & 0.187 & 0.187 & 0.941 & 0.040 \\
& BFLSS    & 0.234 & 0.195 & 0.278 & 0.207 & 0.920 & 0.027 & 0.126 & 0.213 & 0.198 & 0.169 & 0.934 & 0.029 \\
& BFN & \textbf{0.092} & 0.164 & \textbf{0.097} & 0.190 & \textbf{0.991} & 0.047 & \textbf{0.113} & 0.171 & \textbf{0.110} & 0.198 & \textbf{0.962} & 0.043 \\
\hline

\end{tabular}
}
\end{table}

\begin{table}[htbp]
\centering
\caption{Average MSE,  PSE,  and $P_B$ in Case 2 (over 100 trials)}
\label{tab:case2}
\renewcommand{\arraystretch}{1.15}
\resizebox{\textwidth}{!}{%
\begin{tabular}{c l c c c c c c c c c c c c}
\hline
& & \multicolumn{6}{c}{$\rho = 0$} & \multicolumn{6}{c}{$\rho = 0.5$} \\
\cline{3-8} \cline{9-14}
$n$ & Method
& MSE & (sd) & PSE & (sd) & $P_B$ & (sd)
& MSE & (sd) & PSE & (sd) & $P_B$ & (sd) \\
\hline

\multirow{4}{*}{50}
& FL       & 1.939 & 0.409 & 2.077 & 0.419 & 0.790 & 0.031 & 1.826 & 0.447 & 1.862 & 0.296 & 0.749 & 0.027 \\
& BFL      & 1.941 & 0.401 & 2.061 & 0.349 & 0.464 & 0.069 & 2.279 & 0.401 & 2.192 & 0.376 & 0.324 & 0.016 \\
& BFLH     & \textbf{1.600} & 0.245 & 1.700 & 0.272 & 0.441 & 0.071 & 1.832 & 0.262 & 1.963 & 0.229 & 0.429 & 0.032 \\
& BFN & 1.609 & 0.267 & \textbf{1.482} & 0.313 & \textbf{0.822} & 0.019 & \textbf{1.592} & 0.089 & \textbf{1.671} & 0.352 & \textbf{0.840} & 0.018 \\
\hline

\multirow{4}{*}{70}
& FL       & 1.960 & 0.197 & 2.047 & 0.423 & 0.791 & 0.030 & 1.854 & 0.233 & 1.998 & 0.353 & 0.728 & 0.017 \\
& BFL      & 1.915 & 0.343 & 1.882 & 0.358 & 0.509 & 0.033 & 1.698 & 0.431 & 1.650 & 0.203 & 0.578 & 0.031 \\
& BFLH     & 1.758 & 0.462 & 2.004 & 0.433 & 0.673 & 0.043 & \textbf{1.563} & 0.259 & 1.696 & 0.424 & 0.456 & 0.074 \\
& BFN & \textbf{1.654} & 0.353 & \textbf{1.564} & 0.361 & \textbf{0.928} & 0.013 & 1.736 & 0.264 & \textbf{1.596} & 0.162 & \textbf{0.899} & 0.015 \\
\hline

\end{tabular}%
}
\end{table}

\begin{table}[htbp]
\centering
\caption{Average MSE,  PSE,  and $P_B$ in Case 3 (over 100 trials)}
\label{tab:case3}
\renewcommand{\arraystretch}{1.15}
\resizebox{\textwidth}{!}{%
\begin{tabular}{c l c c c c c c c c c c c c}
\hline
& & \multicolumn{6}{c}{$\rho = 0$} & \multicolumn{6}{c}{$\rho = 0.5$} \\
\cline{3-8} \cline{9-14}
$n$ & Method
& MSE & (sd) & PSE & (sd) & $P_B$ & (sd)
& MSE & (sd) & PSE & (sd) & $P_B$ & (sd) \\
\hline

\multirow{4}{*}{100}
& FL       & 3.144 & 0.369 & 3.237 & 0.390 & 0.433 & 0.057 & 0.405 & 0.585 & 0.489 & 0.446 & 0.792 & 0.053 \\
& BFL      & 3.259 & 0.523 & 3.412 & 0.316 & 0.486 & 0.095 & 0.362 & 0.522 & 0.507 & 0.551 & 0.815 & 0.016 \\
& BFLH     & 3.408 & 0.412 & 3.587 & 0.289 & 0.833 & 0.026 & 0.330 & 0.323 & 0.395 & 0.390 & 0.658 & 0.057 \\
& BFN & \textbf{3.049} & 0.293 & \textbf{2.991} & 0.205 & \textbf{0.922} & 0.040 & \textbf{0.320} & 0.046 & \textbf{0.223} & 0.317 & \textbf{0.863} & 0.033 \\
\hline

\multirow{4}{*}{200}
& FL       & 2.890 & 0.190 & 5.811 & 0.607 & 0.573 & 0.066 & 0.278 & 0.156 & 0.252 & 0.451 & 0.812 & 0.018 \\
& BFL      & 2.736 & 0.230 & 2.526 & 0.337 & 0.677 & 0.109 & \textbf{0.216} & 0.446 & 0.284 & 0.212 & 0.892 & 0.011 \\
& BFLH     & 2.176 & 0.298 & 2.264 & 0.386 & 0.879 & 0.035 & 0.514 & 0.188 & 0.823 & 0.578 & 0.888 & 0.010 \\
& BFN & \textbf{2.167} & 0.160 & \textbf{2.088} & 0.540 & \textbf{0.911} & 0.018 & 0.303 & 0.393 & \textbf{0.239} & 0.117 & \textbf{0.899} & 0.019 \\
\hline

\end{tabular}%
}
\end{table}

\subsection{Application}
\label{subsec:real_data_analysis}

To investigate the effectiveness of BFN on real data,  we analyzed the cookie data based on near-infrared spectroscopy \citep{osborne1984biscuit}. 
This dataset consists of 700 reflectance spectra measured at 2 nm intervals from 1100 nm to 2498 nm for 72 cookie dough samples,  along with the compositional information of each sample. 
In this study,  only the flour content was used as the response variable,  while the reflectance spectra corresponding to the 300 wavelengths from 1200 nm to 2400 nm were adopted as explanatory variables. 
In addition,  the 23rd and 61st observations were removed as outliers,  following the outlier information provided with the dataset.

In this dataset, the wavelengths have a natural ordered structure, and there also exists strong correlation among adjacent wavelengths. 
For this reason, methods that perform only variable selection may select only a subset from a group of neighboring wavelengths, resulting in unstable and fragmented estimation results. 
In contrast, BFN aims to select effective wavelengths while simultaneously fusing coefficients for adjacent wavelengths having similar roles. 
Therefore, this dataset provides an appropriate real-data example for examining the effectiveness of BFN, which simultaneously performs variable selection and variable fusion.

In the analysis,  after removing the outliers,  the remaining 70 observations were randomly divided into training and test sets, consisting of 39 and 31 observations, respectively. This random split was repeated 10 times, and the results were summarized over the 10 repetitions.
We fitted BFN, FL, BFL, and BFLH to the training data. 
Note that BFLSS was not included in the analysis because it did not yield stable results. 
For evaluation, we computed the prediction mean squared error based on the test data and recorded the average model dimension. 
In this way, we examined the usefulness of BFN in the real data analysis, with a particular focus on predictive performance and model simplicity.


\begin{table}[htbp]
\centering
\caption{Comparison of the methods for the cookie data}
\label{tab:cookie}
\renewcommand{\arraystretch}{1.2}
\setlength{\tabcolsep}{18pt}
\begin{tabular}{lccc}
\hline
Method & Test MSE & sd & Average model dimension \\
\hline
FL       & 2.509 & 0.262 & 48.2 \\
BFL      & 2.346 & 0.306 & 47.0 \\
BFLH     & 2.438 & 0.272 & 34.9 \\
BFN & 2.415 & 0.313 & 30.2 \\
\hline
\end{tabular}
\end{table}
 
Table~\ref{tab:cookie} presents the comparison results for the methods. 
Although BFN does not attain the smallest test MSE, its value is comparable to those of BFL and BFLH, indicating that its predictive performance is not substantially inferior to that of the existing methods. 
In contrast, BFN achieves the smallest average model dimension and yields the most parsimonious model. 
This result suggests that BFN can perform model selection with fewer variables without substantially sacrificing predictive performance. 
Therefore, for this dataset, BFN shows a favorable balance between prediction accuracy and model simplicity.

\section{Conclusion}

We considered a Bayesian linear regression modeling that simultaneously achieves variable selection and variable fusion for ordered covariates. 
To this end,  we adopted a Bayesian linear regression with the spike-and-slab prior. 
As the slab prior,  we presented the fusion-pMOM prior according to the non-local prior. 
Some theoretical properties of the proposed method were established. 
Monte Carlo simulations and a real data analysis showed the usefulness of the proposed method.

It is of interest to further clarify the theoretical properties of the proposed method under high-dimensional asymptotics: $p$ increases with the sample size $n$. 
Incorporating variable fusion made it possible to theoretically capture the improvement in estimation accuracy for the differences between adjacent regression coefficients. 
However,  it has not been sufficiently described how the estimation accuracy of the individual regression coefficients is improved from the convergence rate in this study.
To clarify the effect of the proposed method on coefficient estimation in more detail,  more refined asymptotic and non-asymptotic analyses are important.
While this study considered only fusion between variables adjacent in a one-dimensional order,  it would be possible to extend the proposed method to more general graph structures or multidimensional adjacency relations.  
Improving computational efficiency and refining hyperparameter selection methods are also important for enhancing the applicability of the proposed method to larger-scale practical problems. 
We leave these for future research.
\section*{Acknowledgments}
A. O. was supported by JSPS KAKENHI Grant Number JP25K24377.
S. K. was supported by JSPS KAKENHI Grant Numbers JP23K11008,  JP23H00809,  and JP23H03352. 
The super-computing resource was provided by Human Genome Center, the Institute of Medical Science, the University of Tokyo.

\bibliographystyle{apalike}
\bibliography{refs}

@inproceedings{wu2021variable,
  title     = {Variable fusion for bayesian linear regression via spike-and-slab priors},
  author    = {Wu, Shengyi and Shimamura, Kaito and Yoshikawa, Kohei and Murayama, Kazuaki and Kawano, Shuichi},
  booktitle = {Proc. 13th KES Int. Conf. Intell. Decis. Technol.},
  pages     = {491--501},
  year      = {2021},
  publisher = {Springer}
}

@article{rossell2017nonlocal,
  title     = {Nonlocal priors for high-dimensional estimation},
  author    = {Rossell, David and Telesca, Donatello},
  journal   = {J. Am. Stat. Assoc.},
  volume    = {112},
  number    = {517},
  pages     = {254--265},
  year      = {2017},
  publisher = {Taylor \& Francis}
}

@article{johnson2010use,
  title     = {On the use of non-local prior densities in Bayesian hypothesis tests},
  author    = {Johnson, Valen E. and Rossell, David},
  journal   = {J. R. Stat. Soc. Ser. B. Stat. Methodol.},
  volume    = {72},
  number    = {2},
  pages     = {143--170},
  year      = {2010},
  publisher = {Oxford University Press}
}

@article{kakikawa2023bayesian,
  title     = {Bayesian fused lasso modeling via horseshoe prior},
  author    = {Kakikawa, Yuko and Shimamura, Kaito and Kawano, Shuichi},
  journal   = {Jpn. J. Stat. Data Sci.},
  volume    = {6},
  number    = {2},
  pages     = {705--727},
  year      = {2023},
  publisher = {Springer}
}

@article{Shi2019MassNonlocal,
  title     = {Model selection using mass-nonlocal prior},
  author    = {Shi, Guiling and Lim, Chae Young and Maiti, Tapabrata},
  journal   = {Stat. Probab. Lett.},
  volume    = {147},
  pages     = {36--44},
  year      = {2019},
  doi       = {10.1016/j.spl.2018.11.027}
}

@article{Kyung2010BayesFusedLasso,
  title     = {Penalized regression, standard errors, and Bayesian lassos},
  author    = {Kyung, Minjung and Gill, Jeff and Ghosh, Malay and Casella, George},
  journal   = {Bayesian Anal.},
  volume    = {5},
  number    = {2},
  pages     = {369--412},
  year      = {2010},
  doi       = {10.1214/10-BA607}
}

@article{gabriel2002structure,
  title     = {The structure of haplotype blocks in the human genome},
  author    = {Gabriel, Stacey B. and Schaffner, Stephen F. and Nguyen, Huy and Moore, Jamie M. and Roy, Jessica and Blumenstiel, Brendan and Higgins, John and DeFelice, Matthew and Lochner, Amy and Faggart, Maura and Liu-Cordero, Shau Neen and Rotimi, Charles and Adeyemo, Adebowale and Cooper, Richard and Ward, Ryk and Lander, Eric S. and Daly, Mark J. and Altshuler, David},
  journal   = {Science},
  volume    = {296},
  number    = {5576},
  pages     = {2225--2229},
  year      = {2002},
  publisher = {Am. Assoc. Adv. Sci.},
  doi       = {10.1126/science.1069424}
}

@article{mcinnes2021gwaspgx,
  title     = {Genomewide Association Studies in Pharmacogenomics},
  author    = {McInnes, Gregory and Yee, Sook Wah and Pershad, Yash and Altman, Russ B.},
  journal   = {Clin. Pharmacol. Ther.},
  volume    = {110},
  number    = {3},
  pages     = {637--648},
  year      = {2021},
  publisher = {Wiley},
  doi       = {10.1002/cpt.2349}
}

@techreport{land1996variablefusion,
  title       = {Variable fusion: a new method of adaptive signal regression},
  author      = {Land, S. and Friedman, J.},
  institution = {Department of Statistics, Stanford University},
  type        = {Technical Report},
  year        = {1996}
}

@article{tibshirani2005fusedlasso,
  title     = {Sparsity and smoothness via the fused lasso},
  author    = {Tibshirani, Robert and Saunders, Michael and Rosset, Saharon and Zhu, Ji and Knight, Keith},
  journal   = {J. R. Stat. Soc. Ser. B. Stat. Methodol.},
  volume    = {67},
  number    = {1},
  pages     = {91--108},
  year      = {2005},
  publisher = {Oxford University Press},
  doi       = {10.1111/j.1467-9868.2005.00490.x}
}

@incollection{rossell2013hdclassifiers,
  title     = {High-Dimensional Bayesian Classifiers Using Non-Local Priors},
  author    = {Rossell, David and Telesca, Donatello and Johnson, Valen E.},
  booktitle = {Stat. Models Data Anal. XV},
  publisher = {Springer},
  pages     = {305--314},
  year      = {2013}
}

@article{kass1995bayes,
  title     = {Bayes Factors},
  author    = {Kass, Robert E. and Raftery, Adrian E.},
  journal   = {J. Am. Stat. Assoc.},
  year      = {1995},
  volume    = {90},
  number    = {430},
  pages     = {773--795},
  doi       = {10.1080/01621459.1995.10476572}
}

@article{hoeting1999bma,
  title     = {Bayesian Model Averaging: A Tutorial},
  author    = {Hoeting, Jennifer A. and Madigan, David and Raftery, Adrian E. and Volinsky, Chris T.},
  journal   = {Stat. Sci.},
  year      = {1999},
  volume    = {14},
  number    = {4},
  pages     = {382--417},
  doi       = {10.1214/ss/1009212519}
}

@article{lindley1957statistical,
  title     = {A Statistical Paradox},
  author    = {Lindley, D. V.},
  journal   = {Biometrika},
  year      = {1957},
  volume    = {44},
  number    = {1-2},
  pages     = {187--192},
  doi       = {10.1093/biomet/44.1-2.187}
}

@article{bartlett1957comment,
  title     = {A Comment on D. V. Lindley's Statistical Paradox},
  author    = {Bartlett, M. S.},
  journal   = {Biometrika},
  year      = {1957},
  volume    = {44},
  number    = {3-4},
  pages     = {533--534},
  doi       = {10.1093/biomet/44.3-4.533}
}

@article{berger1987testing,
  title     = {Testing Precise Hypotheses},
  author    = {Berger, James O. and Delampady, Mohan},
  journal   = {Stat. Sci.},
  year      = {1987},
  volume    = {2},
  number    = {3},
  pages     = {317--335},
  doi       = {10.1214/ss/1177013238}
}

@article{clyde2004model,
  title     = {Model Uncertainty},
  author    = {Clyde, Merlise A. and George, Edward I.},
  journal   = {Stat. Sci.},
  year      = {2004},
  volume    = {19},
  number    = {1},
  pages     = {81--94},
  doi       = {10.1214/088342304000000035}
}

@book{fan2020statistical,
  title     = {Statistical foundations of data science},
  author    = {Fan, Jianqing and Li, Runze and Zhang, Cun-Hui and Zou, Hui},
  year      = {2020},
  publisher = {Chapman and Hall/CRC}
}

@book{buhlmann2016handbook,
  title     = {Handbook of big data},
  author    = {B{\"u}hlmann, Peter and Drineas, Petros and Kane, Michael and van der Laan, Mark},
  year      = {2016},
  publisher = {CRC Press}
}

@article{george1993variable,
  title     = {Variable Selection Via Gibbs Sampling},
  author    = {George, Edward I. and McCulloch, Robert E.},
  journal   = {J. Am. Stat. Assoc.},
  volume    = {88},
  number    = {423},
  pages     = {881--889},
  year      = {1993}
}

@article{walker1969asymptotic,
  title     = {On the Asymptotic Behavior of Posterior Distributions},
  author    = {Walker, A. M.},
  journal   = {J. R. Stat. Soc. Ser. B Methodol.},
  volume    = {31},
  number    = {1},
  pages     = {80--88},
  year      = {1969}
}

@techreport{dawid1999trouble,
  title       = {The Trouble With Bayes Factors},
  author      = {Dawid, A. P.},
  year        = {1999},
  institution = {University College London},
  address     = {London},
  type        = {Technical Report}
}

@article{osborne1984biscuit,
  author    = {Osborne, Brian G. and Fearn, Tom and Miller, A. R. and Douglas, S.},
  title     = {Application of near infrared reflectance spectroscopy to the compositional analysis of biscuits and biscuit doughs},
  journal   = {J. Sci. Food Agric.},
  year      = {1984},
  volume    = {35},
  number    = {1},
  pages     = {99--105},
  doi       = {10.1002/jsfa.2740350116}
}

@article{rockova2018spike,
  author    = {Ro{\v{c}}kov{\'a}, Veronika and George, Edward I.},
  title     = {The Spike-and-Slab LASSO},
  journal   = {J. Am. Stat. Assoc.},
  year      = {2018},
  volume    = {113},
  number    = {521},
  pages     = {431--444},
  doi       = {10.1080/01621459.2016.1260469}
}

@article{ishwaran2005spike,
  author    = {Ishwaran, Hemant and Rao, J. Sunil},
  title     = {Spike and Slab Variable Selection: Frequentist and Bayesian Strategies},
  journal   = {Ann. Stat.},
  year      = {2005},
  volume    = {33},
  number    = {2},
  pages     = {730--773}
}

@article{Shin2018,
  author  = {Shin, Minsuk and Bhattacharya, Anirban and Johnson, Valen E.},
  title   = {Scalable Bayesian Variable Selection Using Nonlocal Prior Densities in Ultrahigh-Dimensional Settings},
  journal = {Statistica Sinica},
  volume  = {28},
  number  = {2},
  pages   = {1053--1078},
  year    = {2018}
}

@misc{miyake2026supplement,
  author = {Miyake, Junya and Okazaki, Akira and Kawano, Shuichi},
  title = {Supplementary Material for "Variable Fusion and Selection via a
Spike-and-Slab Approach with Nonlocal Priors"},
  year = {2026},
  note = {DOI: to be provided by the typesetter}
}

\appendix
\renewcommand{\theequation}{\thesection\arabic{equation}}

\section{Impossibility of a Uniform Prior under Componentwise Independence}
\label{app:uniform_prior_impossible}
A proof of Proposition~\ref{prop:unavailable} is as follows.

\begin{proof}
It suffices to consider the two-dimensional case.
For the two-dimensional regression problem with two explanatory variables,  let $\mathcal{M}_2$ denote the model space that simultaneously accounts for variable selection and variable fusion. We define the mapping
\begin{align}
\phi:\{-1, 0\}\times\{-1, 0, 1\}\to\mathcal{M}_2
\end{align}
that maps an indicator vector to the corresponding model.
Assume a product-form prior on the components of $\bm{\delta}$,  and define
\begin{align}
q^{(j)}_s\coloneqq \Pr(\delta_j=s)\qquad (j=1, 2, \ s\in\{-1, 0, 1\}).
\end{align}

Define the five models by their representatives:
\begin{align*}
m_{00}\coloneqq \phi(0, 0), \quad
m_{10}\coloneqq \phi(-1, 0), \quad
m_{01}\coloneqq \phi(0, -1), \quad
m_{\mathrm{F}}\coloneqq \phi(-1, 1), \quad
m_{\mathrm{S}}\coloneqq \phi(-1, -1).
\label{eq:def_models_by_rep}
\end{align*}
Their preimages are given by
\begin{align}
\phi^{-1}(m_{00})&=\{(0, 0), (0, 1)\},  \quad
\phi^{-1}(m_{10})=\{(-1, 0)\},  \quad
\phi^{-1}(m_{01})=\{(0, -1)\},  \\
\phi^{-1}(m_{\mathrm{F}})&=\{(-1, 1)\},  \quad
\phi^{-1}(m_{\mathrm{S}})=\{(-1, -1)\}. \notag\label{eq:preimageS}
\end{align}
Therefore,  the induced prior on the model space,  $\pi(m)\coloneqq\Pr(\phi(\bm{\delta})=m)$,  satisfies
\begin{align}
\pi(m_{00})&=q^{(1)}_{0}\{q^{(2)}_{0}+q^{(2)}_{1}\},  \quad
\pi(m_{10})=q^{(1)}_{-1}q^{(2)}_{0},  \quad
\pi(m_{01})=q^{(1)}_{0}q^{(2)}_{-1},  \\
\pi(m_{\mathrm{F}})&=q^{(1)}_{-1}q^{(2)}_{1},  \quad
\pi(m_{\mathrm{S}})=q^{(1)}_{-1}q^{(2)}_{-1}. \notag\label{eq:piS}
\end{align}

Now suppose that the prior on $\mathcal{M}_2$ is uniform,  that is, 
\begin{align}
\pi(m_{00})=\pi(m_{10})=\pi(m_{01})=\pi(m_{\mathrm{F}})=\pi(m_{\mathrm{S}}).
\end{align}
From these equalities,  we obtain $q^{(2)}_{0}=q^{(2)}_{1}=q^{(2)}_{-1}$ and $q^{(1)}_{0}=q^{(1)}_{-1}$. Thus, 
\begin{align}
\pi(m_{00})
=q^{(1)}_{0}\{q^{(2)}_{0}+q^{(2)}_{1}\}
=2q^{(1)}_{0}q^{(2)}_{0}, 
\quad
\pi(m_{10})
=q^{(1)}_{0}q^{(2)}_{0}, 
\end{align}
which contradicts $\pi(m_{00})=\pi(m_{10})$.
Hence,  under a componentwise independent prior on $\bm{\delta}$,  it is impossible to assign equal prior probability to all models in $\mathcal{M}_2$.
\end{proof}

\section{A Closed-Form Yielding Equal Prior Probabilities}
\label{app:uniform_prior}

A proof of Proposition~\ref{prop:uniform_objective_prior} is as follows.
\begin{proof}
Let $\mathcal{M}_p$ denote the set of admissible models under consideration. Equivalently,  via the one-to-one correspondence with the indicator vector,  $\mathcal{M}_p$ can be represented as the set of $\bm{\delta}=(\delta_1, \ldots, \delta_p)^\top\in\{-1, 0\}\times\{-1, 0, 1\}^{p-1}$ satisfying
\begin{align}
(\delta_{j-1}, \delta_j)\neq(0, 1), 
\qquad j=2, \ldots, p.
\label{eq:admissible_constraint_app}
\end{align}
Let $(F_n)_{n\ge 1}$ be the Fibonacci sequence defined by
\begin{align}
F_1=F_2=1, \qquad F_{n+2}=F_{n+1}+F_n.
\label{eq:fib_def_app}
\end{align}

For $r\ge 2$ and $a\in\{-1, 0, 1\}$,  define
\begin{align}
N_r(a)
&\coloneqq
\#\Bigl\{(\delta_{1}, \ldots, \delta_{r})\in\{-1, 0, 1\}^r:
\delta_r=a, \ 
(\delta_{s-1}, \delta_s)\neq(0, 1)\ \text{for } s=2, \ldots, r
\Bigr\}. \label{eq:Nr_def_app}
\end{align}
For convenience,  set
$
N_1(a)
\coloneqq 1$.
The constraint \eqref{eq:admissible_constraint_app} implies that,  for $r\ge 2$, 
\begin{align}
N_r(-1) &= N_{r-1}(-1)+N_{r-1}(0)+N_{r-1}(1),  \nonumber\\
N_r(0)  &= N_{r-1}(-1)+N_{r-1}(0),  \label{eq:Nr_rec_app}\\
N_r(1)  &= N_{r-1}(-1)+N_{r-1}(0)+N_{r-1}(1). \nonumber
\end{align}
By symmetry,  $N_r(-1)=N_r(1)$ for all $r$. Let $U_r\coloneqq N_r(-1)=N_r(1)$ and
$V_r\coloneqq N_r(0)$. Then $U_1=V_1=1$ and \eqref{eq:Nr_rec_app} becomes
\begin{align}
U_r &= 2U_{r-1}+V_{r-1},  \nonumber\\
V_r &= U_{r-1}+V_{r-1}. \label{eq:UV_rec_app}
\end{align}
Using the equations in \eqref{eq:fib_def_app},  one verifies by induction that
\begin{align}
U_r = F_{2r+2}, 
\qquad
V_r = F_{2r+1}, 
\qquad r\ge 1.
\label{eq:Nr_closed_app}
\end{align}

For $j=2, \ldots, p$,  recall $r_j=p-j+1$. The parameters can be written as
\begin{align}
\Pr(\delta_j=b\mid \delta_{j-1}=a)
=
\begin{cases}
\dfrac{N_{r_j-1}(b)}{N_{r_j}(a)},  & \text{if }(a, b)\neq(0, 1), \\[6pt]
0,  & \text{if }(a, b)=(0, 1), 
\end{cases}
\qquad a, b\in\{-1, 0, 1\}.
\label{eq:trans_by_count_app}
\end{align}
Moreover,  \eqref{eq:pi_closed_app} is equivalent to
\begin{align}
\Pr(\delta_1=b)
=
\frac{N_{p-1}(b)}{\sum_{c\in\{-1, 0, 1\}}N_{p-1}(c)}, 
\qquad b\in\{-1, 0\}, 
\label{eq:init_by_count_app}
\end{align}
since \eqref{eq:Nr_closed_app} implies
\begin{align}
\sum_{c\in\{-1, 0, 1\}}N_{p-1}(c)
=2F_{2p}+F_{2p-1}
=F_{2p+2}.
\label{eq:sum_N_app}
\end{align}

Fix any $\bm{\delta}\in\mathcal{M}_p$. By the chain rule and \eqref{eq:trans_by_count_app}, 
\begin{align}
\Pr(\bm{\delta})
=
\Pr(\delta_1)\prod_{j=2}^p \Pr(\delta_j\mid \delta_{j-1}) 
=
\frac{N_{p-1}(\delta_1)}{\sum_{c\in\{-1, 0, 1\}}N_{p-1}(c)}
\prod_{j=2}^p
\frac{N_{r_j-1}(\delta_j)}{N_{r_j}(\delta_{j-1})}.
\label{eq:telescoping_app_1}
\end{align}
Since $r_2=p-1$ and $r_{j-1}-1=r_j$ for $j=3, \ldots, p$,  the product telescopes,  and hence
\begin{align}
\Pr(\bm{\delta})
=
\frac{1}{\sum_{c\in\{-1, 0, 1\}}N_{p-1}(c)}
=
\frac{1}{F_{2p+2}}.
\label{eq:telescoping_app_2}
\end{align}
Therefore,  $\Pr(\bm{\delta})$ is constant over $\bm{\delta}\in\mathcal{M}_p$,  proving the claim.
\end{proof}

\section{Normalizing Constant}
\label{app:normalize_constant}

In this appendix,  we derive \eqref{eq:normalizing_constant}.
Let $\bm{Z}=(Z_1, \dots, Z_p)^\top$ be a $p$-dimensional standard normal random vector,  that is, 
$\bm{Z}\sim \mathrm{N}(\bm{0}, I_p)$.
Define the sequence $\{\Phi_k\}_{k\ge 2}$ by
\begin{align}
\Phi_k
&:= \mathbb{E}\!\left[
\left\{\prod_{j=1}^{k} Z_j^2\right\}
\left\{\prod_{j=2}^{k} (Z_j-Z_{j-1})^2\right\}
\right]. \label{eq:Phi_def}
\end{align}
Note that the normalizing constant $C_p$ in \eqref{eq:normalizing_constant} satisfies
$C_p=\Phi_p^{-1}$.
For use in the subsequent calculations,  we also define the sequence
$\{\Psi_k\}_{k\ge 2}$ by
\begin{align}
\Psi_k
&:= \mathbb{E}\!\left[
\left\{\prod_{j=1}^{k-1} Z_j^2\right\} Z_k^4
\left\{\prod_{j=2}^{k} (Z_j-Z_{j-1})^2\right\}
\right].
\end{align}
Furthermore,  let $\mu_k:=\mathbb{E}[Z^k]$ denote the $k$th moment of a standard normal random variable
$Z\sim\mathrm{N}(0, 1)$.
Then,  $\mu_2=1$,  $\mu_4=3$,  and $\mu_6=15$,  while all odd-order moments vanish.

The sequences $\{\Phi_k\}_{k\ge 2}$ and $\{\Psi_k\}_{k\ge 2}$ satisfy the following recursion, 
\begin{align}
\Phi_k
&= \mathbb{E}\!\left[\left\{Z_k^2(Z_k-Z_{k-1})^2\right\}
\left\{\prod_{j=1}^{k-1} Z_j^2\right\}
\left\{\prod_{j=2}^{k-1} (Z_j-Z_{j-1})^2\right\}
\right]
= \mu_4\Phi_{k-1}+\mu_2\Psi_{k-1}, \\
\Psi_k
&= \mathbb{E}\!\left[\left\{Z_k^4(Z_k-Z_{k-1})^2\right\}
\left\{\prod_{j=1}^{k-1} Z_j^2\right\}
\left\{\prod_{j=2}^{k-1} (Z_j-Z_{j-1})^2\right\}
\right]
= \mu_6\Phi_{k-1}+\mu_4\Psi_{k-1}.
\end{align}
Solving these recursion yields the closed-form expression
\begin{align}
\Phi_p
= \frac{(3+\sqrt{15})^{p}-(3-\sqrt{15})^{p}}{2\sqrt{15}}. \label{eq:Phi_closed_form}
\end{align}
Therefore,  the normalizing constant is given by
\begin{align}
C_p
= \Phi_p^{-1}
= \frac{2\sqrt{15}}{(3+\sqrt{15})^{p}-(3-\sqrt{15})^{p}}. \label{eq:Cp_closed_form}
\end{align}

\section{Proofs of the Theoretical Properties}
\label{app:proofs_of_theoretical_properties}

\subsection{Proof of Proposition~\ref{thm:BF_convergence}}

To establish Proposition~\ref{thm:BF_convergence}, we first present the following lemma.
\begin{lemma}[Local convergence rate of the posterior mode obtained from a Taylor expansion around the MLE]\label{lem:mode_rate_from_mle_expansion}
Suppose that the same assumptions as in
Proposition~\ref{prop:sparse_component_mean_local} are satisfied.
Under a fixed model $\bm{\delta}_k$, let
$(\hat{\bm{\theta}}, \hat{\phi})$ denote the MLE of the likelihood, and let
$(\tilde{\bm{\theta}}, \tilde{\phi})$ denote the posterior mode of the
log-posterior distribution. Also, define
$\tilde{\Delta}_i^-:=\tilde{\theta}_i-\tilde{\theta}_{i-1}$ and
$\tilde{\Delta}_i^+:=\tilde{\theta}_{i+1}-\tilde{\theta}_i$.

Then, the following statements hold.
\begin{itemize}
\item[(i)] Suppose that $\theta_i^*\neq 0$. Moreover, suppose that
$\theta_i^*-\theta_{i-1}^*\neq 0$ if $i\in\Lambda$, and that
$\theta_{i+1}^*-\theta_i^*\neq 0$ if $i+1\in\Lambda$.
Then, $\tilde{\theta}_i-\hat{\theta}_i=O_p(n^{-1})$.

\item[(ii)] Suppose that $\theta_i^*=0$. Moreover, suppose that
$\theta_i^*-\theta_{i-1}^*\neq 0$ if $i\in\Lambda$, and that
$\theta_{i+1}^*-\theta_i^*\neq 0$ if $i+1\in\Lambda$.
Then, $\tilde{\theta}_i=O_p(n^{-1/2})$ and
$\tilde{\theta}_i-\hat{\theta}_i=O_p(n^{-1/2})$.

\item[(iii)] Suppose that $\theta_i^*=0$. Moreover, suppose that exactly one of
the two adjacent differences is equal to zero; that is, either
$i\in\Lambda$, $\theta_i^*-\theta_{i-1}^*=0$, $i+1\in\Lambda$, and
$\theta_{i+1}^*-\theta_i^*\neq 0$, or
$i\in\Lambda$, $\theta_i^*-\theta_{i-1}^*\neq 0$, $i+1\in\Lambda$, and
$\theta_{i+1}^*-\theta_i^*=0$.
Then, $\tilde{\theta}_i=O_p(n^{-1/3})$ and
$\tilde{\theta}_i-\hat{\theta}_i=O_p(n^{-1/3})$. Moreover, the adjacent
difference that converges to zero is also $O_p(n^{-1/3})$.

\item[(iv)] Suppose that $\theta_i^*=0$. Moreover, suppose that
$i\in\Lambda$, $i+1\in\Lambda$,
$\theta_i^*-\theta_{i-1}^*=0$, and
$\theta_{i+1}^*-\theta_i^*=0$.
Then, $\tilde{\theta}_i=O_p(n^{-1/4})$ and
$\tilde{\theta}_i-\hat{\theta}_i=O_p(n^{-1/4})$. Moreover,
$\tilde{\Delta}_i^-=O_p(n^{-1/4})$ and $\tilde{\Delta}_i^+=O_p(n^{-1/4})$.
\end{itemize}
\end{lemma}

\begin{proof}
By taking the second-order Taylor expansion of the log-likelihood part around
the MLE $(\hat{\bm{\theta}}, \hat{\phi})$, while keeping the log-prior terms
explicit, differentiating the resulting log-posterior with respect to
$\theta_i$, and setting it equal to zero at the posterior mode
$(\tilde{\bm{\theta}}, \tilde{\phi})$, we obtain
\begin{align}
0
=
-\sum_{j=1}^{p_k} h_{ij}(\tilde{\theta}_j-\hat{\theta}_j)
+\frac{2}{\tilde{\theta}_i}
+\frac{2\mathbf{1}_{\{i\in\Lambda\}}}{\tilde{\Delta}_i^-}
-\frac{2\mathbf{1}_{\{i+1\in\Lambda\}}}{\tilde{\Delta}_i^+}
-\frac{\tilde{\theta}_i}{\tau\tilde{\phi}}
+r_{n, i}.
\label{eq:FOC_theta_i}
\end{align}
where $r_{n, i}=o_p(1)$.
Moreover, under the same assumptions as in Proposition~\ref{thm:BF_convergence},
we have $\sum_{j=1}^{p_k} h_{ij}/n=O_p(1)$.

In case (i), $\tilde{\theta}_i\overset{p}{\to}\theta_i^*\neq 0$.
Therefore, the second, third, fourth, and fifth terms on the right-hand side
of \eqref{eq:FOC_theta_i} are all $O_p(1)$.
Hence,
\begin{align}
\sum_{j=1}^{p_k} h_{ij}(\tilde{\theta}_j-\hat{\theta}_j)=O_p(1)
\end{align}
holds.
Dividing both sides by $n$ yields
\begin{align}
\frac{1}{n}\sum_{j=1}^{p_k} h_{ij}(\tilde{\theta}_j-\hat{\theta}_j)=O_p(n^{-1}).
\end{align}
Therefore, by the regularity conditions, we obtain
\begin{align}
\tilde{\theta}_i-\hat{\theta}_i=O_p(n^{-1}).
\end{align}

In case (ii),  $\tilde{\theta}_i\overset{p}{\to}0$. 
Multiplying both sides of the above equation by $\tilde{\theta}_i$,  we obtain
\begin{align}
\label{eq:Zest}
0
&=
-\tilde{\theta}_i\sum_{j=1}^{p_k} h_{ij}(\tilde{\theta}_j-\hat{\theta}_j)
+2
+\frac{2\mathbf{1}_{\{i\in\Lambda\}}\tilde{\theta}_i}{\tilde{\Delta}_i^-}
-\frac{2\mathbf{1}_{\{i+1\in\Lambda\}}\tilde{\theta}_i}{\tilde{\Delta}_i^+}
-\frac{\tilde{\theta}_i^2}{\tau\tilde{\phi}}
+\tilde{\theta}_i r_{n, i}.
\end{align}
Since all terms after the second term on the right-hand side are $O_p(1)$,  
\begin{align}
\tilde{\theta}_i\sum_{j=1}^{p_k} h_{ij}(\tilde{\theta}_j-\hat{\theta}_j)=O_p(1)
\end{align}
holds. 
Dividing both sides by $n$ gives
\begin{align}
\tilde{\theta}_i\cdot \frac{1}{n}\sum_{j=1}^{p_k} h_{ij}(\tilde{\theta}_j-\hat{\theta}_j)=O_p(n^{-1}).
\end{align}
Therefore,   we obtain $\tilde{\theta}_i=O_p(n^{-1/2})$ and $\tilde{\theta}_i-\hat{\theta}_i=O_p(n^{-1/2})$.

For case (iii),   it is sufficient to show the case where $i\in\Lambda$,  $\theta_i^*-\theta_{i-1}^*=0$,  $i+1\in\Lambda$,  and $\theta_{i+1}^*-\theta_i^*\neq 0$. In this case,  $\tilde{\theta}_i\overset{p}{\to}0$,  $\tilde{\Delta}_i^-\overset{p}{\to}0$,  and $\tilde{\Delta}_i^+=\Theta_p(1)$. Therefore,  multiplying both sides of equation \eqref{eq:Zest} by $\tilde{\theta}_i\tilde{\Delta}_i^-$,  we obtain
\begin{align}
0
&=
-\tilde{\theta}_i\tilde{\Delta}_i^-\sum_{j=1}^{p_k} h_{ij}(\tilde{\theta}_j-\hat{\theta}_j)
+2\tilde{\Delta}_i^-
+2\tilde{\theta}_i
-\frac{2\mathbf{1}_{\{i+1\in\Lambda\}}\tilde{\theta}_i\tilde{\Delta}_i^-}{\tilde{\Delta}_i^+}
-\frac{\tilde{\theta}_i^2\tilde{\Delta}_i^-}{\tau\tilde{\phi}}
+\tilde{\theta}_i\tilde{\Delta}_i^- r_{n,  i}.
\end{align}
Since all terms after the second term on the right-hand side are $O_p(1)$,  
\begin{align}
\tilde{\theta}_i\tilde{\Delta}_i^-\sum_{j=1}^{p_k} h_{ij}(\tilde{\theta}_j-\hat{\theta}_j)=O_p(1)
\end{align}
holds. Dividing both sides by $n$ gives
\begin{align}
\tilde{\theta}_i\tilde{\Delta}_i^-\cdot \frac{1}{n}\sum_{j=1}^{p_k} h_{ij}(\tilde{\theta}_j-\hat{\theta}_j)=O_p(n^{-1}).
\end{align}
Here,  $\sum_{j=1}^{p_k} h_{ij}(\tilde{\theta}_j-\hat{\theta}_j)/n$ has the same order as $\tilde{\theta}_i-\hat{\theta}_i$,  and $\tilde{\theta}_i$,  $\tilde{\Delta}_i^-$,  and $\tilde{\theta}_i-\hat{\theta}_i$ have the same convergence rate. Therefore,  we obtain $\tilde{\theta}_i=O_p(n^{-1/3})$,  $\tilde{\theta}_i-\hat{\theta}_i=O_p(n^{-1/3})$,  and $\tilde{\Delta}_i^-=O_p(n^{-1/3})$. 

In case (iv),  we have $\tilde{\theta}_i\overset{p}{\to}0$,  $\tilde{\Delta}_i^-\overset{p}{\to}0$,  and $\tilde{\Delta}_i^+\overset{p}{\to}0$. Therefore,  multiplying both sides of equation \eqref{eq:Zest} by $\tilde{\theta}_i\tilde{\Delta}_i^-\tilde{\Delta}_i^+$,  we obtain
\begin{align}
0
&=
-\tilde{\theta}_i\tilde{\Delta}_i^-\tilde{\Delta}_i^+\sum_{j=1}^{p_k} h_{ij}(\tilde{\theta}_j-\hat{\theta}_j)
+2\tilde{\Delta}_i^-\tilde{\Delta}_i^+
+2\tilde{\theta}_i\tilde{\Delta}_i^+
-2\tilde{\theta}_i\tilde{\Delta}_i^-
-\frac{\tilde{\theta}_i^2\tilde{\Delta}_i^-\tilde{\Delta}_i^+}{\tau\tilde{\phi}}
+\tilde{\theta}_i\tilde{\Delta}_i^-\tilde{\Delta}_i^+ r_{n,  i}.
\end{align}
Since all terms after the second term on the right-hand side are $O_p(1)$,  
\begin{align}
\tilde{\theta}_i\tilde{\Delta}_i^-\tilde{\Delta}_i^+\sum_{j=1}^{p_k} h_{ij}(\tilde{\theta}_j-\hat{\theta}_j)=O_p(1)
\end{align}
holds. Dividing both sides by $n$ gives
\begin{align}
\tilde{\theta}_i\tilde{\Delta}_i^-\tilde{\Delta}_i^+\cdot \frac{1}{n}\sum_{j=1}^{p_k} h_{ij}(\tilde{\theta}_j-\hat{\theta}_j)=O_p(n^{-1}).
\end{align}
Here,  $(1/n)\sum_{j=1}^{p_k} h_{ij}(\tilde{\theta}_j-\hat{\theta}_j)$ has the same order as $\tilde{\theta}_i-\hat{\theta}_i$,  and moreover $\tilde{\theta}_i$,  $\tilde{\Delta}_i^-$,  $\tilde{\Delta}_i^+$,  and $\tilde{\theta}_i-\hat{\theta}_i$ have the same convergence rate. Therefore,  we obtain $\tilde{\theta}_i=O_p(n^{-1/4})$,  $\tilde{\theta}_i-\hat{\theta}_i=O_p(n^{-1/4})$,  $\tilde{\Delta}_i^-=O_p(n^{-1/4})$,  and $\tilde{\Delta}_i^+=O_p(n^{-1/4})$. 
\end{proof}
The following is the proof of Proposition~3.
\begin{proof}

We approximate the marginal likelihood of each model by Laplace approximation in a neighborhood of each posterior mode. First,  the marginal likelihood $m_t(y_n)$ of the true model $\bm{\delta}_t$ is given by
\begin{align}
m_t(y_n)
&=
\int
e^{L_n(\bm{\theta}_t, \phi_t)}
Q_t(\bm{\theta}_t)
(2\pi\tau\phi_t)^{-p_t/2}
\exp\!\left(
-\frac{1}{2\tau\phi_t}\bm{\theta}_t^\top\bm{\theta}_t
\right)
\pi(\phi_t)
\, d\bm{\theta}_t\, d\phi_t .\label{eq:marginal_likkelihood}
\end{align}
Here,  \(L_n(\bm{\theta}_k, \phi_k)\) denotes the log-likelihood under model \(\bm{\delta}_k\),  evaluated at \((\bm{\theta}_k, \phi_k)\) based on the sample \(\bm{y}\),  and \(Q_k(\bm{\theta}_k)\) denotes the non-local part of the slab density under model \(\bm{\delta}_k\). The parameter space \(\mathbb{R}^{p_k+1}\) is partitioned by the hyperplanes \(\theta_{k_j}=0\) \((j=1, \ldots, p_k)\) and \(\theta_{k_j}-\theta_{k_{j-1}}=0\) \((j=2, \ldots, p_k)\). Let \(D_k\) denote the number of resulting regions. Suppose that,  in each such region,  a posterior mode appears,  and denote these modes by \((\tilde{\bm{\theta}}_k^{(m)}, \tilde{\phi}_k^{(m)})\),  \(m=1, \ldots, D_k\). Among them,  let \((\tilde{\bm{\theta}}_k^{(1)}, \tilde{\phi}_k^{(1)})\) denote the posterior mode lying in the same region as the MLE \((\hat{\bm{\theta}}, \hat{\phi})\). Then,  the marginal likelihood under the true model in \eqref{eq:marginal_likkelihood} can be approximated by Laplace's method around each of these modes as follows.

Let $(\tilde{\bm{\theta}}_t^{(1)}, \tilde{\phi}_t^{(1)})$ denote the mode lying in the same neighborhood as the MLE. Hence, 
\begin{align}
m_t(y_n)
&\approx
\sum_{m=1}^{D_t}
e^{L_n(\tilde{\bm{\theta}}_t^{(m)}, \tilde{\phi}_t^{(m)})}
Q_t(\tilde{\bm{\theta}}_t^{(m)})
(2\pi\tau\tilde{\phi}_t^{(m)})^{-p_t/2}
\exp\!\left(
-\frac{1}{2\tau\tilde{\phi}_t^{(m)}}(\tilde{\bm{\theta}}_t^{(m)})^\top\tilde{\bm{\theta}}_t^{(m)}
\right)
\nonumber\\
&\qquad\qquad\qquad\times
\pi(\tilde{\phi}_t^{(m)})
\left|
H_t(\tilde{\bm{\theta}}_t^{(m)}, \tilde{\phi}_t^{(m)})
\right|^{-1/2}, 
\label{eq:BF_mt_laplace}
\end{align}
where  $H_t$ is the Hessian of the negative log-posterior.

Now,  for the mode $(\tilde{\bm{\theta}}_t^{(1)}, \tilde{\phi}_t^{(1)})$,  the Walker condition implies that
$n^{-1}H_t(\tilde{\bm{\theta}}_t^{(1)}, \tilde{\phi}_t^{(1)})\overset{p}{\to}J_t$
for some positive definite matrix $J_t$.
Therefore, 
$\left|H_t(\tilde{\bm{\theta}}_t^{(1)}, \tilde{\phi}_t^{(1)})\right|^{-1/2}
=
O_p(n^{-p_t/2})$.
Moreover, 
$Q_t(\tilde{\bm{\theta}}_t^{(1)})=O_p(1)$ holds.
Hence,  the Laplace approximation at the mode $(\tilde{\bm{\theta}}_t^{(1)}, \tilde{\phi}_t^{(1)})$ can be written as
\begin{align}
e^{L_n(\tilde{\bm{\theta}}_t^{(1)}, \tilde{\phi}_t^{(1)})}
\, O_p(n^{-p_t/2}).
\label{eq:BF_mt_main}
\end{align}

On the other hand,  for modes other than this dominant one,  at least one coefficient or one difference approaches zero,  so that polynomial decay arises. Furthermore,  by Condition A,  their log-likelihood values decay exponentially relative to that of the dominant mode. Therefore,  the contributions of the other modes are negligible compared with that of the dominant mode. Hence, 
\begin{align}
m_t(y_n)
=
e^{L_n(\tilde{\bm{\theta}}_t^{(1)}, \tilde{\phi}_t^{(1)})}
\, O_p(n^{-p_t/2}), 
\label{eq:BF_mt_final}
\end{align}
is obtained.

First,  consider the case $\bm{\delta}_t \subset \bm{\delta}_k$. In this case,  $\bm{\delta}_k$ is an overfitted model,  and its marginal likelihood $m_k(y_n)$ can also be evaluated by Laplace approximation in a neighborhood of each posterior mode. Focusing on the terms corresponding to the modes,  the Bayes factor can be written as
\begin{align}
\mathrm{BF}_{kt}
=
\frac{m_k(y_n)}{m_t(y_n)}, 
\end{align}
and can be evaluated as a ratio of the corresponding factors. Here,  for an overfitted model,  the log-likelihood difference remains of order $O_p(1)$,  and hence
\begin{align}
e^{L_n(\tilde{\bm{\theta}}_k^{(m)}, \tilde{\phi}_k^{(m)})-L_n(\tilde{\bm{\theta}}_t^{(1)}, \tilde{\phi}_t^{(1)})}=O_p(1)
\end{align}
holds.
Moreover,  the non-local term satisfies
\begin{align}
Q_k(\tilde{\bm{\theta}}_k^{(m)})= O_p(n^{-(p_k-p_t)}).
\end{align}
In addition,  by the bounded continuity of $\pi(\phi)$ and the regularity of the Hessian,  both the prior distribution for $\phi$ and the Hessian ratio are of order $O_p(1)$.
Therefore, 
\begin{align}
\mathrm{BF}_{kt}
&=
O_p(1)\cdot O_p\!\left(n^{-(p_k-p_t)}\right)\cdot
n^{-\frac{1}{2}(p_k-p_t)}\cdot O_p(1)
\nonumber\\
&=
O_p\!\left(n^{-\frac{3}{2}(p_k-p_t)}\right).
\end{align}

Finally,  consider the case $\bm{\delta}_t\not\subset \bm{\delta}_k$. By Condition A,  we have $\mathrm{KL}(\bm{\delta}_t, \bm{\delta}_k)>0$ in this case. Therefore,  for each posterior mode $(\tilde{\bm{\theta}}_k^{(m)}, \tilde{\phi}_k^{(m)})$, 
\begin{align}
L_n(\tilde{\bm{\theta}}_k^{(m)}, \tilde{\phi}_k^{(m)})
-
L_n(\tilde{\bm{\theta}}_t^{(1)}, \tilde{\phi}_t^{(1)})
=
-n\, \mathrm{KL}(\bm{\delta}_t, \bm{\delta}_k)+o_p(n)
\end{align}
holds. Hence,  there exists some $C>0$ such that
\begin{align}
\exp\!\left(
L_n(\tilde{\bm{\theta}}_k^{(m)}, \tilde{\phi}_k^{(m)})
-
L_n(\tilde{\bm{\theta}}_t^{(1)}, \tilde{\phi}_t^{(1)})
\right)
=
O_p(e^{-Cn}).
\end{align}
On the other hand,  the ratio of the non-local terms,  the prior ratio,  the Hessian ratio,  and the dimensional factor in the Laplace approximation are at most of polynomial order,  so that the exponential term is dominant. Therefore, 
\begin{align}
\mathrm{BF}_{kt}
=
O_p(e^{-Cn})
\qquad \text{for some } C>0
\end{align}
is obtained. This completes the proof.
\end{proof}

\subsection{Proof of Proposition~\ref{prop:sparse_component_mean_local}}

The proof of Proposition~\ref{prop:sparse_component_mean_local} is as follows. 

\begin{proof}
For simplicity,   we omit the model index throughout. 
Define
\begin{align}
\ell_n(\bm{\theta},  \phi)
:=
L_n(\bm{\theta},  \phi)
+
\log p(\bm{\theta}\mid \phi)
+
\log \pi(\phi),  
\label{eq:ell_n_final}
\end{align}
and let $g_n(\bm{\theta},  \phi):=-n^{-1}\ell_n(\bm{\theta},  \phi)$. 
Then
\begin{align}
E(\theta_i\mid \bm{y})
=
\frac{
\int\!\!\int \theta_i \exp\{-n g_n(\bm{\theta},  \phi)\}\,  d\bm{\theta}\,  d\phi
}{
\int\!\!\int \exp\{-n g_n(\bm{\theta},  \phi)\}\,  d\bm{\theta}\,  d\phi
}. 
\label{eq:posterior_mean_ratio_final}
\end{align}

Set $\bm{z}:=(z_1,  \ldots,  z_{p_k+1})^\top:=(\theta_1,  \ldots,  \theta_{p_k},  \phi)^\top$ and
$\tilde{\bm{z}}:=(\tilde{\theta}_1,  \ldots,  \tilde{\theta}_{p_k},  \tilde{\phi})^\top$. 
For $1\le r,  s,  t\le p_k+1$,   define
\begin{align}
a_{rs}
&:=
\frac{\partial^2 g_n(\bm{z})}{\partial z_r\,  \partial z_s}
\Bigg|_{\bm{z}=\tilde{\bm{z}}},   \quad
t_{rst}
:=
\frac{\partial^3 g_n(\bm{z})}{\partial z_r\,  \partial z_s\,  \partial z_t}
\Bigg|_{\bm{z}=\tilde{\bm{z}}}. 
\label{eq:derivative_notation_final}
\end{align}
Let $A:=(a_{rs})_{1\le r,  s\le p_k+1}$ be the Hessian matrix of $g_n$ at $\tilde{\bm{z}}$,   and write $A^{-1}=(b_{rs})_{1\le r,  s\le p_k+1}$. 

By the Laplace expansion for posterior means around the main mode,  
\begin{align}
E(\theta_i\mid \bm{y})
=
\tilde{\theta}_i
+
\frac{1}{n}
\sum_{j=1}^{p_k+1}
b_{ij}
\left(
-\frac12\sum_{r,  s=1}^{p_k+1} b_{rs}\,  t_{rsj}
\right)
+
O_p(n^{-2}). 
\label{eq:laplace_mean_final}
\end{align}
Thus it is enough to evaluate the third derivatives $t_{rsj}$. 

Write $\Delta_i:=\theta_i-\theta_{i-1}$ and $\Delta_{i+1}:=\theta_{i+1}-\theta_i$. 
In the $\theta_i$-direction,   the singular part of the fusion-pMOM log-prior is
\begin{align}
q_i(\bm{\theta},  \phi)
&:=
2\log|\theta_i|
+
2\,  1_{\{i\in\Lambda\}}\log|\theta_i-\theta_{i-1}|
+
2\,  1_{\{i+1\in\Lambda\}}\log|\theta_{i+1}-\theta_i|
-
\frac{\theta_i^2}{2\tau\phi}. 
\label{eq:qi_final}
\end{align}
Its third derivative is
\begin{align}
\frac{\partial^3 q_i}{\partial\theta_i^3}
&=
\frac{4}{\theta_i^3}
+
\frac{4}{(\theta_i-\theta_{i-1})^3}\,  1_{\{i\in\Lambda\}}
-
\frac{4}{(\theta_{i+1}-\theta_i)^3}\,  1_{\{i+1\in\Lambda\}}. 
\label{eq:third_derivative_final}
\end{align}

First,   suppose that $\theta_i^*\neq 0$. 
Then Lemma~\ref{lem:mode_rate_from_mle_expansion} yields $\tilde{\theta}_i-\hat{\theta}_i=O_p(n^{-\frac{1}{2}})$,   and hence $\tilde{\theta}_i=\theta_i^*+O_p(n^{-1})=O_p(1)$. 
Moreover,   if $i\in\Lambda$ then $\theta_i^*-\theta_{i-1}^*\neq 0$,   and if $i+1\in\Lambda$ then $\theta_{i+1}^*-\theta_i^*\neq 0$,   so $\tilde{\Delta}_i=O_p(1)$ and $\tilde{\Delta}_{i+1}=O_p(1)$. 
Therefore \eqref{eq:third_derivative_final} gives
\begin{align}
\frac{\partial^3 q_i}{\partial\theta_i^3}
\Bigg|_{(\bm{\theta},  \phi)=(\tilde{\bm{\theta}},  \tilde{\phi})}
=
O_p(1). 
\label{eq:q_third_fixed_final}
\end{align}
Hence the relevant third derivatives of $g_n$ are $O_p(n^{-1})$,   while all remaining third derivatives are also $O_p(n^{-1})$ by the Walker condition and local regularity. 
Moreover,   the normalized Hessian converges to a finite positive-definite matrix,   so $b_{rs}=O_p(1)$ for all $r,  s$. 
Thus \eqref{eq:laplace_mean_final} implies
\begin{align}
E(\theta_i\mid \bm{y})-\tilde{\theta}_i
=
O_p(n^{-2}). 
\label{eq:mean_minus_mode_fixed_final}
\end{align}
Combining this with $\tilde{\theta}_i=\theta_i^*+O_p(n^{-1})$,   we obtain
$E(\theta_i\mid \bm{y})=\theta_i^*+O_p(n^{-1})=\theta_i^*+O_p(n^{-1/2})$. 
This proves \eqref{eq:mean_rate_fixed_final}. 

Next,   suppose that $\theta_i^*=0$ and that both neighboring true differences are nonzero. 
Then Lemma~\ref{lem:mode_rate_from_mle_expansion} yields $\tilde{\theta}_i=O_p(n^{-1/2})$ and $\tilde{\theta}_i-\hat{\theta}_i=O_p(n^{-1/2})$. 
Moreover,   $\tilde{\Delta}_i=O_p(1)$ and $\tilde{\Delta}_{i+1}=O_p(1)$. 
Hence \eqref{eq:third_derivative_final} gives
\begin{align}
\frac{\partial^3 q_i}{\partial\theta_i^3}
\Bigg|_{(\bm{\theta},  \phi)=(\tilde{\bm{\theta}},  \tilde{\phi})}
=
O_p(n^{3/2}). 
\label{eq:q_third_zero_isolated_final}
\end{align}
Since $g_n=-n^{-1}\ell_n$,   the corresponding third derivatives of $g_n$ satisfy
$t_{rst}=O_p(n^{1/2})$ for the terms involving the $\theta_i$-direction,   while all remaining third derivatives are $O_p(n^{-1/2})$. 
Hence \eqref{eq:laplace_mean_final} implies
\begin{align}
E(\theta_i\mid \bm{y})-\tilde{\theta}_i
=
O_p(n^{-1/2}). 
\label{eq:mean_minus_mode_zero_isolated_final}
\end{align}
Combining this with $\tilde{\theta}_i=O_p(n^{-1/2})$,   we obtain
$E(\theta_i\mid \bm{y})=O_p(n^{-1/2})$. 
This proves \eqref{eq:mean_rate_zero_isolated_final}. 

Now suppose that $\theta_i^*=0$ and exactly one neighboring true difference is equal to zero. 
For definiteness,   assume that $\theta_i^*-\theta_{i-1}^*=0$ and $\theta_{i+1}^*-\theta_i^*\neq 0$. 
Then Lemma~\ref{lem:mode_rate_from_mle_expansion} yields
$\tilde{\theta}_i=O_p(n^{-1/3})$,  
$\tilde{\theta}_i-\hat{\theta}_i=O_p(n^{-1/3})$,  
and $\tilde{\Delta}_i=O_p(n^{-1/3})$,  
whereas $\tilde{\Delta}_{i+1}=O_p(1)$. 
Hence \eqref{eq:third_derivative_final} gives
\begin{align}
\frac{\partial^3 q_i}{\partial\theta_i^3}
\Bigg|_{(\bm{\theta},  \phi)=(\tilde{\bm{\theta}},  \tilde{\phi})}
=
O_p(n). 
\label{eq:q_third_zero_onefusion_final}
\end{align}
Therefore the relevant third derivatives of $g_n$ are $O_p(1)$,   and \eqref{eq:laplace_mean_final} gives
\begin{align}
E(\theta_i\mid \bm{y})-\tilde{\theta}_i
=
O_p(n^{-1}). 
\label{eq:mean_minus_mode_zero_onefusion_final}
\end{align}
Since $\tilde{\theta}_i=O_p(n^{-1/3})$,   it follows that
$E(\theta_i\mid \bm{y})=O_p(n^{-1/3})$. 
This proves \eqref{eq:mean_rate_zero_onefusion_final}. 
The case $\theta_i^*-\theta_{i-1}^*\neq 0$ and $\theta_{i+1}^*-\theta_i^*=0$ is entirely analogous. 

Finally,   suppose that $\theta_i^*=0$,   $\theta_i^*-\theta_{i-1}^*=0$,   and $\theta_{i+1}^*-\theta_i^*=0$. 
Then Lemma~\ref{lem:mode_rate_from_mle_expansion} yields
$\tilde{\theta}_i=O_p(n^{-1/4})$,  
$\tilde{\theta}_i-\hat{\theta}_i=O_p(n^{-1/4})$,  
$\tilde{\Delta}_i=O_p(n^{-1/4})$,  
and $\tilde{\Delta}_{i+1}=O_p(n^{-1/4})$. 
Hence \eqref{eq:third_derivative_final} yields
\begin{align}
\frac{\partial^3 q_i}{\partial\theta_i^3}
\Bigg|_{(\bm{\theta},  \phi)=(\tilde{\bm{\theta}},  \tilde{\phi})}
=
O_p(n^{3/4}). 
\label{eq:q_third_zero_twofusion_final}
\end{align}
Therefore the relevant third derivatives of $g_n$ are $O_p(n^{-1/4})$,   and \eqref{eq:laplace_mean_final} gives
\begin{align}
E(\theta_i\mid \bm{y})-\tilde{\theta}_i
=
O_p(n^{-5/4}). 
\label{eq:mean_minus_mode_zero_twofusion_final}
\end{align}
Since $\tilde{\theta}_i=O_p(n^{-1/4})$,   we conclude that
$E(\theta_i\mid \bm{y})=O_p(n^{-1/4})$. 
This proves \eqref{eq:mean_rate_zero_twofusion_final}. 

Thus \eqref{eq:mean_rate_fixed_final}--\eqref{eq:mean_rate_zero_twofusion_final} hold. 
\end{proof}

\subsection{Proof of Theorem~\ref{prop:expectation}}
\begin{proof}
Decompose the posterior mean as
\begin{align}
\E(\theta_i\mid \bm y_n)
&=
\E(\theta_i\mid \bm{\delta}_t,  \bm y_n)P(\bm{\delta}_t\mid \bm y_n)
\nonumber\\
&\quad+
\sum_{k:\,  \bm{\delta}_t\subset \bm{\delta}_k}
\E(\theta_i\mid \bm{\delta}_k,  \bm y_n)P(\bm{\delta}_k\mid \bm y_n)
+
\sum_{k:\,  \bm{\delta}_t\not\subset \bm{\delta}_k}
\E(\theta_i\mid \bm{\delta}_k,  \bm y_n)P(\bm{\delta}_k\mid \bm y_n). 
\label{eq:bma_decomp}
\end{align}

First,   under the true model $\bm{\delta}_t$,   Proposition~\ref{prop:sparse_component_mean_local} implies that,   when $\theta_i^*\neq 0$,  
\begin{align}
\E(\theta_i\mid \bm{\delta}_t,  \bm y_n)
=
\theta_i^*+O_p(n^{-1/2})
\label{eq:true_model_nonzero}
\end{align}
holds. 
On the other hand,   when $\theta_i^*=0$,   that coefficient is not included in the true model,   and therefore
\begin{align}
\E(\theta_i\mid \bm{\delta}_t,  \bm y_n)=0
\label{eq:true_model_zero}
\end{align}
holds. 

Next,   for an overfitted model $\bm{\delta}_k \supset \bm{\delta}_t$,   Proposition 3 implies that
\begin{align}
P(\bm{\delta}_k\mid \bm y_n)
=
O_p\!\left(
n^{-\frac32(p_k-p_t)}
\right)
\qquad (\bm{\delta}_t\subset \bm{\delta}_k)
\label{eq:overfit_prob}
\end{align}
holds. 
Moreover,   Proposition~\ref{prop:sparse_component_mean_local} yields,   according to the local configuration around $\theta_i^*=0$,  
\begin{align}
&\theta_i^*=0,  \
\theta_i^*-\theta_{i-1}^*\neq 0,  \
\theta_{i+1}^*-\theta_i^*\neq 0
\Longrightarrow
\E(\theta_i\mid \bm{\delta}_k,  \bm y_n)=O_p(n^{-1/2}),  
\label{eq:overfit_conditional_isolated}
\\
&\theta_i^*=0,  \
\text{exactly one of }\theta_i^*-\theta_{i-1}^*\text{ and } \theta_{i+1}^*-\theta_i^*\text{ is equal to zero}
\Longrightarrow
\E(\theta_i\mid \bm{\delta}_k,  \bm y_n)=O_p(n^{-1/3}),  
\label{eq:overfit_conditional_onefusion}
\\
&\theta_i^*=0,  \
\theta_i^*-\theta_{i-1}^*=0,  \
\theta_{i+1}^*-\theta_i^*=0
\Longrightarrow
\E(\theta_i\mid \bm{\delta}_k,  \bm y_n)=O_p(n^{-1/4})
\quad (\bm{\delta}_t\subset \bm{\delta}_k). 
\label{eq:overfit_conditional_twofusion}
\end{align}
Hence,   multiplying \eqref{eq:overfit_prob} by \eqref{eq:overfit_conditional_isolated},   \eqref{eq:overfit_conditional_onefusion},   and \eqref{eq:overfit_conditional_twofusion},   respectively,   the dominant case is again $p_k=p_t+1$,   and we obtain
\begin{align}
&\theta_i^*=0,  \
\theta_i^*-\theta_{i-1}^*\neq 0,  \
\theta_{i+1}^*-\theta_i^*\neq 0
\Longrightarrow
\sum_{k:\,  \bm{\delta}_t\subset \bm{\delta}_k}
\E(\theta_i\mid \bm{\delta}_k,  \bm y_n)P(\bm{\delta}_k\mid \bm y_n)
=
O_p(n^{-2}),  
\label{eq:overfit_sum_isolated}
\\
&\theta_i^*=0,  \
\text{exactly one of }\theta_i^*-\theta_{i-1}^*\text{ and } \theta_{i+1}^*-\theta_i^*\text{ is equal to zero} \nonumber \\
& \hspace{50mm} \Longrightarrow
\sum_{k:\,  \bm{\delta}_t\subset \bm{\delta}_k}
\E(\theta_i\mid \bm{\delta}_k,  \bm y_n)P(\bm{\delta}_k\mid \bm y_n)
=
O_p(n^{-11/6}),  
\label{eq:overfit_sum_onefusion}
\\
&\theta_i^*=0,  \
\theta_i^*-\theta_{i-1}^*=0,  \
\theta_{i+1}^*-\theta_i^*=0
\Longrightarrow
\sum_{k:\,  \bm{\delta}_t\subset \bm{\delta}_k}
\E(\theta_i\mid \bm{\delta}_k,  \bm y_n)P(\bm{\delta}_k\mid \bm y_n)
=
O_p(n^{-7/4}). 
\label{eq:overfit_sum_twofusion}
\end{align}

On the other hand,   for an underfitted model $\bm{\delta}_k \not \supset  \bm{\delta}_t$,   Proposition 5 implies that
\begin{align}
P(\bm{\delta}_k\mid \bm y_n)=O_p(e^{-Cn})
\qquad (\bm{\delta}_t\not\subset \bm{\delta}_k)
\label{eq:underfit_prob}
\end{align}
for some $C>0$,   and since the conditional posterior mean is at most of polynomial order,  
\begin{align}
\sum_{k:\,  \bm{\delta}_t\not\subset \bm{\delta}_k}
\E(\theta_i\mid \bm{\delta}_k,  \bm y_n)P(\bm{\delta}_k\mid \bm y_n)
=
O_p(e^{-Cn})
\label{eq:underfit_sum}
\end{align}
follows. 

Substituting these into \eqref{eq:bma_decomp},   when $\theta_i^*\neq 0$,   the contribution from the true model in \eqref{eq:true_model_nonzero} is the leading term,   and the other two terms are smaller.  Hence,  
\begin{align}
\E(\theta_i\mid \bm y_n)
=
\theta_i^*+O_p(n^{-1/2})
\end{align}
is obtained. 

When $\theta_i^*=0$ and both neighboring true differences are nonzero,   the contribution from the true model is zero by \eqref{eq:true_model_zero},   \eqref{eq:overfit_sum_isolated} is the leading term,   and \eqref{eq:underfit_sum} is negligible.  Therefore,  
\begin{align}
\E(\theta_i\mid \bm y_n)
=
O_p(n^{-2}). 
\end{align}

When $\theta_i^*=0$ and exactly one neighboring true difference is equal to zero,   the contribution from the true model is again zero,   \eqref{eq:overfit_sum_onefusion} is the leading term,   and \eqref{eq:underfit_sum} is negligible.  Therefore,  
\begin{align}
\E(\theta_i\mid \bm y_n)
=
O_p(n^{-11/6}). 
\end{align}

Finally,   when $\theta_i^*=0$ and both neighboring true differences are equal to zero,   the contribution from the true model is zero,   \eqref{eq:overfit_sum_twofusion} is the leading term,   and \eqref{eq:underfit_sum} is negligible.  Therefore,  
\begin{align}
\E(\theta_i\mid \bm y_n)
=
O_p(n^{-7/4}). 
\end{align}
\end{proof}

\subsection{Proof of Theorem~\ref{thm:posterior_mean_difference_fixed_dim}}

\begin{proof}
Decompose the posterior mean into three parts:
\begin{align}
E(\Delta_j\mid \bm{y})
&=
E(\Delta_j\mid \bm{\delta}_t,  \bm{y})P(\bm{\delta}_t\mid \bm{y})
+
\sum_{k:\,  \bm{\delta}_t\subset\bm{\delta}_k}
E(\Delta_j\mid \bm{\delta}_k,  \bm{y})P(\bm{\delta}_k\mid \bm{y})
\notag\\
&\qquad
+
\sum_{k:\,  \bm{\delta}_t\not\subset\bm{\delta}_k}
E(\Delta_j\mid \bm{\delta}_k,  \bm{y})P(\bm{\delta}_k\mid \bm{y}). 
\label{eq:three_way_decomp_diff}
\end{align}

Because $p$ is fixed,   the number of models with $p_k-p_t=m$ is finite for each fixed $m$. 
Hence,   by Proposition~\ref{thm:BF_convergence},  
\begin{align}
\sum_{k:\,  \bm{\delta}_t\subset\bm{\delta}_k,  \ p_k-p_t=1}
P(\bm{\delta}_k\mid \bm{y})
&=
O_p(n^{-3/2}),  
\label{eq:one_extra_total_prob_diff}
\\
\sum_{k:\,  \bm{\delta}_t\subset\bm{\delta}_k,  \ p_k-p_t\ge2}
P(\bm{\delta}_k\mid \bm{y})
&=
O_p(n^{-3}),  
\label{eq:two_or_more_extra_total_prob_diff}
\\
\sum_{k:\,  \bm{\delta}_t\not\subset\bm{\delta}_k}
P(\bm{\delta}_k\mid \bm{y})
&=
O_p(e^{-Cn})
\label{eq:underfit_total_prob_diff}
\end{align}
for some $C>0$. 
Therefore,  
\begin{align}
P(\bm{\delta}_t\mid \bm{y})
=
1-O_p(n^{-3/2}). 
\label{eq:true_model_post_prob_diff}
\end{align}

We now consider four cases. 

If $\Delta_j^*\neq0$, 
\begin{align}
E(\Delta_j\mid \bm{\delta}_t,  \bm{y})
=
\Delta_j^*+O_p(n^{-1/2}). 
\label{eq:true_model_diff_fixed_use}
\end{align}
Hence \eqref{eq:true_model_post_prob_diff} gives
\begin{align}
E(\Delta_j\mid \bm{\delta}_t,  \bm{y})P(\bm{\delta}_t\mid \bm{y})
&=
\bigl(\Delta_j^*+O_p(n^{-1/2})\bigr)
\bigl(1- O_p(n^{-3/2})\bigr)
\notag\\
&=
\Delta_j^*+O_p(n^{-1/2}). 
\label{eq:true_term_diff_fixed}
\end{align}
For the overfitted models,   $E(\Delta_j\mid \bm{\delta}_k,  \bm{y})=O_p(1)$ by stochastic boundedness,   so \eqref{eq:one_extra_total_prob_diff} and \eqref{eq:two_or_more_extra_total_prob_diff} yield
\begin{align}
\sum_{k:\,  \bm{\delta}_t\subset\bm{\delta}_k}
E(\Delta_j\mid \bm{\delta}_k,  \bm{y})P(\bm{\delta}_k\mid \bm{y})
&=
O_p(1)\cdot  O_p(n^{-3/2})
+
O_p(1)\cdot O_p(n^{-3})
\notag\\
&=
O_p(n^{-3/2}). 
\label{eq:overfit_term_diff_fixed}
\end{align}
For the underfitted models,   stochastic boundedness and \eqref{eq:underfit_total_prob_diff} imply
\begin{align}
\sum_{k:\,  \bm{\delta}_t\not\subset\bm{\delta}_k}
E(\Delta_j\mid \bm{\delta}_k,  \bm{y})P(\bm{\delta}_k\mid \bm{y})
=
O_p(e^{-Cn}). 
\label{eq:underfit_term_diff_fixed}
\end{align}
Combining \eqref{eq:true_term_diff_fixed},   \eqref{eq:overfit_term_diff_fixed},   and \eqref{eq:underfit_term_diff_fixed},   we obtain
\begin{align}
E(\Delta_j\mid \bm{y})
=
\Delta_j^*+O_p(n^{-1/2}),  
\end{align}
which proves \eqref{eq:thm_diff_fixed_nonzero}. 

Next,   suppose that $\Delta_j^*=0$,   $\theta_j^*-\theta_{j-1}^*\neq 0$,   and $\theta_{j+1}^*-\theta_j^*\neq 0$. 
Then,   under the true model $\bm{\delta}_t$,  we obatin
\begin{align}
E(\Delta_j\mid \bm{\delta}_t,  \bm{y})=0. 
\label{eq:true_model_diff_zero_exact_isolated}
\end{align}
For the overfitted models with $p_k-p_t=1$,   the assumed conditional mean rate gives
$E(\Delta_j\mid \bm{\delta}_k,  \bm{y})=O_p(n^{-1/2})$. 
Hence \eqref{eq:one_extra_total_prob_diff} implies
\begin{align}
\sum_{k:\,  \bm{\delta}_t\subset\bm{\delta}_k,  \ p_k-p_t=1}
E(\Delta_j\mid \bm{\delta}_k,  \bm{y})P(\bm{\delta}_k\mid \bm{y})
&=
O_p(n^{-1/2})\cdot O_p(n^{-3/2})
\notag\\
&=
O_p(n^{-2}). 
\label{eq:one_extra_zero_diff_isolated}
\end{align}
For the overfitted models with $p_k-p_t\ge2$,   stochastic boundedness and \eqref{eq:two_or_more_extra_total_prob_diff} give
\begin{align}
\sum_{k:\,  \bm{\delta}_t\subset\bm{\delta}_k,  \ p_k-p_t\ge2}
E(\Delta_j\mid \bm{\delta}_k,  \bm{y})P(\bm{\delta}_k\mid \bm{y})
=
O_p(n^{-3}). 
\label{eq:multi_extra_zero_diff_isolated}
\end{align}
For the underfitted models,  we obtain
\begin{align}
\sum_{k:\,  \bm{\delta}_t\not\subset\bm{\delta}_k}
E(\Delta_j\mid \bm{\delta}_k,  \bm{y})P(\bm{\delta}_k\mid \bm{y})
=
O_p(e^{-Cn}). 
\label{eq:underfit_zero_diff_isolated}
\end{align}
Substituting \eqref{eq:true_model_diff_zero_exact_isolated},   \eqref{eq:one_extra_zero_diff_isolated},   \eqref{eq:multi_extra_zero_diff_isolated},   and \eqref{eq:underfit_zero_diff_isolated} into \eqref{eq:three_way_decomp_diff},   we obtain
\begin{align}
E(\Delta_j\mid \bm{y})
=
O_p(n^{-2}),  
\end{align}
which proves \eqref{eq:thm_diff_zero_isolated}. 

Next,   suppose that $\Delta_j^*=0$ and exactly one of $\theta_j^*-\theta_{j-1}^*$ and $\theta_{j+1}^*-\theta_j^*$ is equal to zero. 
Then,   under the true model $\bm{\delta}_t$,  we obtain
\begin{align}
E(\Delta_j\mid \bm{\delta}_t,  \bm{y})=0. 
\label{eq:true_model_diff_zero_exact_onefusion}
\end{align}
For the overfitted models with $p_k-p_t=1$,   the assumed conditional mean rate gives
$E(\Delta_j\mid \bm{\delta}_k,  \bm{y})=O_p(n^{-1/3})$. 
Hence \eqref{eq:one_extra_total_prob_diff} implies
\begin{align}
\sum_{k:\,  \bm{\delta}_t\subset\bm{\delta}_k,  \ p_k-p_t=1}
E(\Delta_j\mid \bm{\delta}_k,  \bm{y})P(\bm{\delta}_k\mid \bm{y})
&=
O_p(n^{-1/3})\cdot O_p(n^{-3/2})
\notag\\
&=
O_p(n^{-11/6}). 
\label{eq:one_extra_zero_diff_onefusion}
\end{align}
For the overfitted models with $p_k-p_t\ge2$,   stochastic boundedness and \eqref{eq:two_or_more_extra_total_prob_diff} give
\begin{align}
\sum_{k:\,  \bm{\delta}_t\subset\bm{\delta}_k,  \ p_k-p_t\ge2}
E(\Delta_j\mid \bm{\delta}_k,  \bm{y})P(\bm{\delta}_k\mid \bm{y})
=
O_p(n^{-3}). 
\label{eq:multi_extra_zero_diff_onefusion}
\end{align}
For the underfitted models,  we obtain
\begin{align}
\sum_{k:\,  \bm{\delta}_t\not\subset\bm{\delta}_k}
E(\Delta_j\mid \bm{\delta}_k,  \bm{y})P(\bm{\delta}_k\mid \bm{y})
=
O_p(e^{-Cn}). 
\label{eq:underfit_zero_diff_onefusion}
\end{align}
Substituting \eqref{eq:true_model_diff_zero_exact_onefusion},   \eqref{eq:one_extra_zero_diff_onefusion},   \eqref{eq:multi_extra_zero_diff_onefusion},   and \eqref{eq:underfit_zero_diff_onefusion} into \eqref{eq:three_way_decomp_diff},   we obtain
\begin{align}
E(\Delta_j\mid \bm{y})
=
 O_p(n^{-11/6}),  
\end{align}
which proves \eqref{eq:thm_diff_zero_onefusion}. 

Finally,   suppose that $\Delta_j^*=0$,   $\theta_j^*-\theta_{j-1}^*=0$,   and $\theta_{j+1}^*-\theta_j^*=0$. 
Then,   under the true model $\bm{\delta}_t$,  we obtain
\begin{align}
E(\Delta_j\mid \bm{\delta}_t,  \bm{y})=0. 
\label{eq:true_model_diff_zero_exact_twofusion}
\end{align}
For the overfitted models with $p_k-p_t=1$,   the assumed conditional mean rate gives
$E(\Delta_j\mid \bm{\delta}_k,  \bm{y})=O_p(n^{-1/4})$. 
Hence \eqref{eq:one_extra_total_prob_diff} implies
\begin{align}
\sum_{k:\,  \bm{\delta}_t\subset\bm{\delta}_k,  \ p_k-p_t=1}
E(\Delta_j\mid \bm{\delta}_k,  \bm{y})P(\bm{\delta}_k\mid \bm{y})
&=
O_p(n^{-1/4})\cdot O_p(n^{-3/2})
\notag\\
&=
O_p(n^{-7/4}). 
\label{eq:one_extra_zero_diff_twofusion}
\end{align}
For the overfitted models with $p_k-p_t\ge2$,   stochastic boundedness and \eqref{eq:two_or_more_extra_total_prob_diff} give
\begin{align}
\sum_{k:\,  \bm{\delta}_t\subset\bm{\delta}_k,  \ p_k-p_t\ge2}
E(\Delta_j\mid \bm{\delta}_k,  \bm{y})P(\bm{\delta}_k\mid \bm{y})
=
O_p(n^{-3}). 
\label{eq:multi_extra_zero_diff_twofusion}
\end{align}
For the underfitted models,  we obtain
\begin{align}
\sum_{k:\,  \bm{\delta}_t\not\subset\bm{\delta}_k}
E(\Delta_j\mid \bm{\delta}_k,  \bm{y})P(\bm{\delta}_k\mid \bm{y})
=
O_p(e^{-Cn}). 
\label{eq:underfit_zero_diff_twofusion}
\end{align}
Substituting \eqref{eq:true_model_diff_zero_exact_twofusion},   \eqref{eq:one_extra_zero_diff_twofusion},   \eqref{eq:multi_extra_zero_diff_twofusion},   and \eqref{eq:underfit_zero_diff_twofusion} into \eqref{eq:three_way_decomp_diff},   we obtain
\begin{align}
E(\Delta_j\mid \bm{y})
=
O_p(n^{-7/4}),  
\end{align}
which proves \eqref{eq:thm_diff_zero_twofusion}.

Thus \eqref{eq:thm_diff_fixed_nonzero}--\eqref{eq:thm_diff_zero_twofusion} hold. 
\end{proof}

\section{Derivation of the Marginal Likelihood and Full Conditional Distributions}
\label{sec:full_conditional}

\subsection{\texorpdfstring{Derivation of the marginal likelihood $m_{\bm{\delta}}(\bm{y})$}{Derivation of the marginal likelihood mdelta(y)}}

It can be derived by the following algebraic calculation.
Fix the model $\bm{\delta}$,  and define
\begin{align*}
Q_{\bm{\delta}}(\bm{\theta}_{\bm{\delta}})
&:=
\prod_{j=1}^{p_{\bm{\delta}}}\theta_{\bm{\delta}_j}^2
\prod_{j\in\Lambda_{\bm{\delta}}}(\theta_{\bm{\delta}_j}-\theta_{\bm{\delta}_{j-1}})^2, 
&
A_{\bm{\delta}}
&:=
X_{\bm{\delta}}^{\top}X_{\bm{\delta}}+I,  \\
\tilde{\bm{\beta}}_{\bm{\delta}}
&:=
A_{\bm{\delta}}^{-1}X_{\bm{\delta}}^{\top}\bm{y}, 
&
R_{\bm{\delta}}
&:=
\bm{y}^{\top}(I-X_{\bm{\delta}}A_{\bm{\delta}}^{-1}X_{\bm{\delta}}^{\top})\bm{y},  \\
\nu_{\bm{\delta}}
&:=
n+2p_{\bm{\delta}}+2|\Lambda_{\bm{\delta}}|+2\alpha, 
&
s_{\bm{\delta}}^{2}
&:=
\frac{R_{\bm{\delta}}+2\psi}{\nu_{\bm{\delta}}}.
\end{align*}
Then,  the marginal likelihood is written as
\begin{align}
m(\bm{y}\mid\bm{\delta})
&=
\int\int
(2\pi\sigma^2)^{-n/2}
\exp\left(
-\frac{1}{2\sigma^2}
\|\bm{y}-X_{\bm{\delta}}\bm{\theta}_{\bm{\delta}}\|^2
\right)
C_{\bm{\delta}}
Q_{\bm{\delta}}(\bm{\theta}_{\bm{\delta}})
(2\pi)^{-p_{\bm{\delta}}/2}
(\sigma^2)^{-\left(\frac{3}{2}p_{\bm{\delta}}+|\Lambda_{\bm{\delta}}|\right)}
\notag\\
&\qquad\qquad\qquad\times
\exp\left(
-\frac{1}{2\sigma^2}
\bm{\theta}_{\bm{\delta}}^{\top}\bm{\theta}_{\bm{\delta}}
\right)
\frac{\psi^{\alpha}}{\Gamma(\alpha)}
(\sigma^2)^{-(\alpha+1)}
\exp\left(
-\frac{\psi}{\sigma^2}
\right)
\, d\bm{\theta}_{\bm{\delta}}\, d\sigma^2 \\
&=
C_{\bm{\delta}}(2\pi)^{-(n+p_{\bm{\delta}})/2}\frac{\psi^{\alpha}}{\Gamma(\alpha)}
\int
Q_{\bm{\delta}}(\bm{\theta}_{\bm{\delta}})
\int
(\sigma^2)^{-\left(\alpha+1+\frac{n+3p_{\bm{\delta}}}{2}+|\Lambda_{\bm{\delta}}|\right)}
\notag\\
&\qquad\qquad\qquad\times
\exp\left(
-\frac{
2\psi+\|\bm{y}-X_{\bm{\delta}}\bm{\theta}_{\bm{\delta}}\|^2+\bm{\theta}_{\bm{\delta}}^{\top}\bm{\theta}_{\bm{\delta}}
}{2\sigma^2}
\right)
\, d\sigma^2\, d\bm{\theta}_{\bm{\delta}} \\
&=
C_{\bm{\delta}}(2\pi)^{-(n+p_{\bm{\delta}})/2}\frac{\psi^{\alpha}}{\Gamma(\alpha)}
\Gamma\left(\frac{\nu_{\bm{\delta}}+p_{\bm{\delta}}}{2}\right)\notag\\
&\qquad\qquad\qquad\times
\int
Q_{\bm{\delta}}(\bm{\theta}_{\bm{\delta}})
\left(
\frac{
2\psi+\|\bm{y}-X_{\bm{\delta}}\bm{\theta}_{\bm{\delta}}\|^2+\bm{\theta}_{\bm{\delta}}^{\top}\bm{\theta}_{\bm{\delta}}
}{2}
\right)^{-\frac{\nu_{\bm{\delta}}+p_{\bm{\delta}}}{2}}
\, d\bm{\theta}_{\bm{\delta}} \\
&=
C_{\bm{\delta}}(2\pi)^{-(n+p_{\bm{\delta}})/2}\frac{\psi^{\alpha}}{\Gamma(\alpha)}
\Gamma\left(\frac{\nu_{\bm{\delta}}+p_{\bm{\delta}}}{2}\right)
2^{\frac{\nu_{\bm{\delta}}+p_{\bm{\delta}}}{2}}
(\nu_{\bm{\delta}}s_{\bm{\delta}}^{2})^{-\frac{\nu_{\bm{\delta}}+p_{\bm{\delta}}}{2}}\notag\\
&\qquad\qquad\qquad\times
\int
Q_{\bm{\delta}}(\bm{\theta}_{\bm{\delta}})
\left(
1+
\frac{
(\bm{\theta}_{\bm{\delta}}-\tilde{\bm{\beta}}_{\bm{\delta}})^{\top}
A_{\bm{\delta}}
(\bm{\theta}_{\bm{\delta}}-\tilde{\bm{\beta}}_{\bm{\delta}})
}{\nu_{\bm{\delta}}s_{\bm{\delta}}^{2}}
\right)^{-\frac{\nu_{\bm{\delta}}+p_{\bm{\delta}}}{2}}
\, d\bm{\theta}_{\bm{\delta}} \\
&=
C_{\bm{\delta}}(2\pi)^{-(n+p_{\bm{\delta}})/2}\frac{\psi^{\alpha}}{\Gamma(\alpha)}
2^{\frac{\nu_{\bm{\delta}}+p_{\bm{\delta}}}{2}}
\pi^{p_{\bm{\delta}}/2}
|A_{\bm{\delta}}|^{-1/2}
(\nu_{\bm{\delta}}s_{\bm{\delta}}^{2})^{-\nu_{\bm{\delta}}/2}
\Gamma\left(\frac{\nu_{\bm{\delta}}}{2}\right)
E_{t}\left[Q_{\bm{\delta}}(\bm{\theta}_{\bm{\delta}})\right], 
\end{align}
where $E_t$ denotes the expectation with respect to the $p_{\bm{\delta}}$-dimensional $t$ distribution with mean $\tilde{\bm{\beta}}_{\bm{\delta}}$,  scale matrix $s_{\bm{\delta}}^{2}A_{\bm{\delta}}^{-1}$,  and degrees of freedom $\nu_{\bm{\delta}}$.

\subsection{\texorpdfstring{Full conditional distribution of $\delta_j$}{Full conditional distribution of delta}}

The full conditional distribution of $\delta_j$ is given by
\begin{align}
p(\delta_j=x\mid \bm{y}, \bm{\delta}_{-j}, \sigma^2, \{\bm{\omega}_\ell\}, \{\kappa_{\ell, -1}\})
&\propto
p(\bm{y}\mid \bm{\delta}^{(j\leftarrow x)})
p(\bm{\delta}^{(j\leftarrow x)}\mid \{\bm{\omega}_\ell\}, \{\kappa_{\ell, -1}\}), 
\qquad x\in\{-1, 0, 1\}.
\end{align}
Hence, 
\begin{align}
\delta_j\mid \bm{y}, \bm{\delta}_{-j}, \sigma^2, \{\bm{\omega}_\ell\}, \{\kappa_{\ell, -1}\}
\sim
\mathrm{Categorical}(h_{j, -1}, h_{j, 0}, h_{j, 1}), 
\end{align}

\subsection{\texorpdfstring{Full conditional distribution of $\bm{\theta}_{\bm{\delta}}$}{Full conditional distribution of theta}}

The full conditional distribution of $\bm{\theta}_{\bm{\delta}}$ is given by
\begin{align}
p(\bm{\theta}_{\bm{\delta}}\mid \bm{y}, \sigma^2, \bm{\delta})
&\propto
p(\bm{y}\mid \bm{\theta}_{\bm{\delta}}, \sigma^2, \bm{\delta})
p(\bm{\theta}_{\bm{\delta}}\mid \sigma^2, \bm{\delta}) \\
&\propto
Q_{\bm{\delta}}(\bm{\theta}_{\bm{\delta}})
\exp\left(
-\frac{1}{2\sigma^2}
\left(
\|\bm{y}-X_{\bm{\delta}}\bm{\theta}_{\bm{\delta}}\|^2
+
\bm{\theta}_{\bm{\delta}}^{\top}\bm{\theta}_{\bm{\delta}}
\right)
\right) \\
&\propto
Q_{\bm{\delta}}(\bm{\theta}_{\bm{\delta}})
\exp\left(
-\frac{1}{2\sigma^2}
(\bm{\theta}_{\bm{\delta}}-\tilde{\bm{\beta}}_{\bm{\delta}})^{\top}
A_{\bm{\delta}}
(\bm{\theta}_{\bm{\delta}}-\tilde{\bm{\beta}}_{\bm{\delta}})
\right).
\end{align}
Hence, 
\begin{align}
\bm{\theta}_{\bm{\delta}}\mid \bm{y}, \sigma^2, \bm{\delta}
\sim
\mathrm{fusion\text{-}pMOM}
\left(
\tilde{\bm{\beta}}_{\bm{\delta}}, 
\sigma^2A_{\bm{\delta}}^{-1}
\right).
\end{align}

\subsection{\texorpdfstring{Full conditional distribution of $\sigma^2$}{Full conditional distribution of sigma}}

The full conditional distribution of $\sigma^2$ is derived as follows:
\begin{align}
p(\sigma^2\mid \bm{y}, \bm{\theta}_{\bm{\delta}}, \bm{\delta})
&\propto
p(\bm{y}\mid \bm{\theta}_{\bm{\delta}}, \sigma^2, \bm{\delta})
p(\bm{\theta}_{\bm{\delta}}\mid \sigma^2, \bm{\delta})
p(\sigma^2) \\
&\propto
(\sigma^2)^{-n/2}
(\sigma^2)^{-\left(\frac{3}{2}p_{\bm{\delta}}+|\Lambda_{\bm{\delta}}|\right)}
(\sigma^2)^{-(\alpha+1)}
\exp\left(
-\frac{
\|\bm{y}-X_{\bm{\delta}}\bm{\theta}_{\bm{\delta}}\|^2
+
\bm{\theta}_{\bm{\delta}}^{\top}\bm{\theta}_{\bm{\delta}}
+
2\psi
}{2\sigma^2}
\right) \\
&\propto
(\sigma^2)^{-\left(
\alpha+1+\frac{n+3p_{\bm{\delta}}}{2}+|\Lambda_{\bm{\delta}}|
\right)}
\exp\left(
-\frac{
\|\bm{y}-X_{\bm{\delta}}\bm{\theta}_{\bm{\delta}}\|^2
+
\bm{\theta}_{\bm{\delta}}^{\top}\bm{\theta}_{\bm{\delta}}
+
2\psi
}{2\sigma^2}
\right).
\end{align}
Hence, 
\begin{align}
\sigma^2\mid \bm{y}, \bm{\theta}_{\bm{\delta}}, \bm{\delta}
\sim
IG\left(
\alpha+\frac{n+3p_{\bm{\delta}}}{2}+|\Lambda_{\bm{\delta}}|, 
\, 
\psi+\frac{
\|\bm{y}-X_{\bm{\delta}}\bm{\theta}_{\bm{\delta}}\|^2
+
\bm{\theta}_{\bm{\delta}}^{\top}\bm{\theta}_{\bm{\delta}}
}{2}
\right).
\end{align}
\subsection{\texorpdfstring{Full conditional distribution of $\bm{\omega}_j$}{Full conditional distribution of omegaj}}

For each $j=2, \ldots, p$,  let
\[
\bm{\omega}_j
:=
(\omega_{j, -1}, \omega_{j, 0}, \omega_{j, 1})^\top, 
\qquad
\bm{\omega}_j \sim \mathrm{Dirichlet}(Aa_{j, -1}, Aa_{j, 0}, Aa_{j, 1}).
\]
The transition probability of $\delta_j$ is governed by $\bm{\omega}_j$ when $\delta_{j-1}\in\{-1, 1\}$,  whereas it is governed by $\kappa_{j, -1}$ and $\kappa_{j, 0}$ when $\delta_{j-1}=0$. Therefore,  the full conditional distribution of $\bm{\omega}_j$ is derived by considering the following two cases.

If $\delta_{j-1}\in\{-1, 1\}$,  then
\begin{align}
p(\bm{\omega}_j\mid \bm{\delta})
&\propto
p(\delta_j\mid \delta_{j-1}, \bm{\omega}_j)\, p(\bm{\omega}_j) \\
&\propto
\omega_{j, -1}^{\mathbf{1}_{\{\delta_j=-1\}}}
\omega_{j, 0}^{\mathbf{1}_{\{\delta_j=0\}}}
\omega_{j, 1}^{\mathbf{1}_{\{\delta_j=1\}}}
\omega_{j, -1}^{Aa_{j, -1}-1}
\omega_{j, 0}^{Aa_{j, 0}-1}
\omega_{j, 1}^{Aa_{j, 1}-1} \\
&\propto
\omega_{j, -1}^{Aa_{j, -1}+\mathbf{1}_{\{\delta_j=-1\}}-1}
\omega_{j, 0}^{Aa_{j, 0}+\mathbf{1}_{\{\delta_j=0\}}-1}
\omega_{j, 1}^{Aa_{j, 1}+\mathbf{1}_{\{\delta_j=1\}}-1}.
\end{align}
Hence, 
\begin{align}
\bm{\omega}_j\mid \bm{\delta}
\sim
\mathrm{Dirichlet}
\left(
Aa_{j, -1}+\textbf{1}_{\{\delta_j=-1\}}, 
Aa_{j, 0}+\textbf{1}_{\{\delta_j=0\}}, 
Aa_{j, 1}+\textbf{1}_{\{\delta_j=1\}}
\right).
\end{align}

If $\delta_{j-1}=0$,  then the transition probability of $\delta_j$ is determined by $(\kappa_{j, -1}, \kappa_{j, 0})$ and does not involve $\bm{\omega}_j$. Therefore,  the full conditional distribution of $\bm{\omega}_j$ coincides with its prior distribution:
\begin{align}
\bm{\omega}_j\mid \bm{\delta}
\sim
\mathrm{Dirichlet}(Aa_{j, -1}, Aa_{j, 0}, Aa_{j, 1}).
\end{align}

\subsection{\texorpdfstring{Full conditional distribution of $\bm{\kappa}_{j, -1}$}{Full conditional distribution of kappaj, -1}}
If $\delta_{j-1}=0$,  then
\begin{align}
p(\kappa_{j, -1}\mid \bm{\delta})
&\propto
p(\delta_j\mid \delta_{j-1}, \kappa_{j, -1})\, p(\kappa_{j, -1}) \\
&\propto
\kappa_{j, -1}^{\mathbf{1}_{\{\delta_j=-1\}}}
(1-\kappa_{j, -1})^{\mathbf{1}_{\{\delta_j=0\}}}
\kappa_{j, -1}^{Bc_{j, 0}-1}
(1-\kappa_{j, -1})^{Bd_{j, 0}-1} \\
&\propto
\kappa_{j, -1}^{Bc_{j, 0}+\mathbf{1}_{\{\delta_j=-1\}}-1}
(1-\kappa_{j, -1})^{Bd_{j, 0}+\mathbf{1}_{\{\delta_j=0\}}-1}.
\end{align}
Therefore, 
\begin{align}
\kappa_{j, -1}\mid \bm{\delta}
\sim
\mathrm{Beta}
\left(
Bc_{j, 0}+\mathbf{1}_{\{\delta_j=-1\}}, 
Bd_{j, 0}+\mathbf{1}_{\{\delta_j=0\}}
\right).
\end{align}

If $\delta_{j-1}\in\{-1, 1\}$,  then the transition probability of $\delta_j$ does not involve $\kappa_{j, -1}$. Therefore,  the full conditional distribution of $\kappa_{j, -1}$ coincides with its prior distribution:
\begin{align}
\kappa_{j, -1}\mid \bm{\delta}
\sim
\mathrm{Beta}(Bc_{j, 0}, Bd_{j, 0}).
\end{align}

\end{document}